\documentclass[a4paper,USenglish,cleveref,autoref,thm-restate
]{lipics-v2021}

\usepackage{amsmath,amssymb}
\usepackage{mathtools}
\usepackage{booktabs}
\usepackage{tabularx}
\usepackage{xltabular}
\newcolumntype{Y}{>{\raggedright\arraybackslash}X}
\usepackage{stmaryrd}
\usepackage{graphicx}

\nolinenumbers 

\crefname{conjecture}{Conjecture}{Conjectures}
\Crefname{conjecture}{Conjecture}{Conjectures}

\title{Light Cone Consistency: Closure, Ordering, and the Single-Observer Boundary}
\titlerunning{Light Cone Consistency}

\author{Rob Landers}{Swytch BV, Utrecht, The Netherlands}{rob@getswytch.com}{}{}
\author{Kaben Kramer}{Swytch BV, Utrecht, The Netherlands}{kaben@getswytch.com}{}{}

\authorrunning{Landers \& Kramer}

\Copyright{Rob Landers and Kaben Kramer}

\ccsdesc[500]{Theory of computation~Distributed algorithms}
\ccsdesc[300]{Computer systems organization~Distributed architectures}

\keywords{distributed systems, consistency models, causal DAG, fork resolution, light cone consistency, impossibility results}

\begin{document}

    \maketitle

    \begin{abstract}

        Every distributed system is a message-passing system, and every message-passing system is a growing causal DAG observed by a set of observers.
        We treat each observer's consistency as two operators on its visible sub-DAG (a causal-closure filter $C$, fixing which dependencies it must have seen, and a fork resolution $O$, ordering the concurrent forks the filter admits) and give the resulting space the structure the flat catalog of named models lacks.
        The operators are coupled, asymmetrically: an order that refines causality supplies closure its filter never demanded.
        That coupling yields a decidable readability order (which configuration's data another can read honestly) with a factoring dichotomy: the order splits across the $C$ and $O$ axes exactly when ordering does not refine causality, and refuses to when it does, the cross-axis gap being the closure ordering supplies.
        On that order sit a consistency ratchet (a level lost under migration is never regained) and a Detection = Prevention bound: a system can tell its order inverted causality only if it retained exactly what would have prevented the inversion.

        The classical results land at clean coordinates in the same system, not as new claims: resolving a fork demands retaining the causal history that distinguishes its branches (database folklore, here an impossibility for every message-passing system) and linearizability resolves as a composite of two systems, a store and a global real-time serializer supplying an order no single observer's light cone can.
        The named models are configurations of $(C, O)$, exact over the standard-safety fragment and generative past it, predicting configurations the catalog has not named.
        LCC is a formalization of the observer-relative consistency model of Burgess and Gerlits.

    \end{abstract}

    \section{Introduction}
    \label{sec:intro}

    Every distributed system is a message-passing system, and every message-passing system is a
    growing causal directed acyclic graph (DAG). Messages are vertices; causal dependencies are
    edges. Each observer sees a sub-DAG that grows monotonically as messages arrive, and every
    consistency guarantee is a constraint on the shape and ordering of that visible sub-DAG. We
    work in the crash/omission model: Byzantine and shared-memory behaviour are out of scope, and
    real-time agreement \emph{across} observers is not a constraint on one sub-DAG at all but the
    boundary we characterise (\cref{sec:boundary}).

    A consistency model, in this view, is built from two operators on an observer's sub-DAG. The
    first is a causal-closure \emph{filter} $C$: it fixes which causal dependencies of a message
    an observer must have seen before acting on it (none, those sharing an object or a session,
    or all of them). The second is \emph{fork resolution} $O$: within the sub-DAG the filter
    admits, concurrent messages (forks, with no causal order between them) are placed into a
    single order over a message-intrinsic scope. A response function $F$ reports what the
    observer returns, and its waiting behaviour $R$ sets how long it defers before answering. The
    named consistency models are configurations of $(C, O, F)$ (\cref{sec:model}).

    $C$ and $O$ look like independent dials, and as raw settings they are: a system can impose a
    total order over data it never causally relates---a network switch serialises packets that
    way.
    What couples them is causality.
    For a resolved order to \emph{refine} causality, an observer must already hold the dependencies that order respects,
    so a causally faithful $O$ at a given scope entails $C$ at the matching scope.
    The rule that produces the order (delivery order, last-writer-wins, a deterministic merge) is interchangeable and invisible here;
    what matters is whether the order refines causality. An order that does is \emph{$\kappa$-clean};
    one that inverts a causal edge carries a \emph{$\kappa$-scar}. That split, not the choice of
    rule, is where consistency lives (\cref{sec:readability}).

    Four consequences organise the paper. \emph{Retention}: resolving a fork requires retaining
    the causal history that distinguishes its branches, so a system that resolves cannot discard
    that history first---a bound on every message-passing system, not a storage detail
    (\cref{sec:retention}). \emph{Readability}: the configurations form a decidable order (given
    data written under one configuration, exactly which can read it honestly?) with a factoring
    dichotomy: it splits across the $C$ and $O$ axes exactly when ordering does not refine
    causality, and refuses to when it does, the $\kappa$-coupling being that refusal
    (\cref{sec:readability}). \emph{Mergeability}: two divergent views reconcile exactly when
    their orders agree, and the partition tradeoff (availability or ordering) is what that rule
    forces when they cannot (\cref{sec:merge}). \emph{Permanence}: a value migrated through a
    weaker configuration cannot be read back at the stronger one (the consistency ratchet), and a
    $\kappa$-scar can be detected only by a system that retained what would have prevented
    it---Detection\,$=$\,Prevention (\cref{sec:ratchet-scars}).

    The boundary is the sharpest claim.
    Linearizability (a single global order respecting real
    time) is not a single configuration but a composite of two message-passing systems: a store
    holding the data, and a global serializer that orders it and routes clients through it in
    real time (\cref{sec:boundary}). The reason is relativistic. A single order over operations
    that are pairwise concurrent (spacelike-separated) is a \emph{preferred frame}, a simultaneity
    imposed on events that have none intrinsically; no observer's light cone supplies it, so it
    must be manufactured at a distinguished point. This is why everything else here is
    \emph{single-observer} consistency: it is exactly what an observer can establish from its own
    sub-DAG, and linearizability is the standard model that steps outside that frame. LCC is,
    throughout, a formalization of the observer-relative architecture of Burgess and
    Gerlits~\cite{burgess2022}; we place the classical impossibilities in its coordinate system
    but make no claim to unify them.

    \paragraph*{Contributions.}
    \begin{itemize}
        \item The asymmetric coupling of closure and ordering: named consistency as
        configurations of a closure filter $C$ and a fork resolution $O$ on each observer's
        sub-DAG (\cref{sec:model}), where a causally faithful order entails closure
        (\cref{sec:readability})---the cross-axis dependence a flat parameter space cannot
        isolate.
        \item A decidable readability order with a factoring dichotomy (\cref{thm:dichotomy}),
        the consistency ratchet (\cref{thm:ratchet}), and the Detection\,$=$\,Prevention bound
        (\cref{thm:detect-prevent}).
        \item The mergeability frontier (\cref{thm:merge}, \cref{thm:merge-frontier}): two
        divergent views reconcile exactly when their orders agree, with the partition tradeoff
        (\cref{cor:partition}) the choice forced when they cannot.
        \item The retention bound (\cref{thm:retention}): database folklore, here an impossibility
        for every message-passing system.
        \item The single-observer boundary (\cref{thm:globality}, \cref{cor:linearizability}):
        linearizability is a composite of two message-passing systems, the global order-scope it
        requires being irreducible---recovering the classical distinction within the
        single-observer frame, not importing it.
        \item A glossary mapping the configuration space onto the named models
        (\cref{tab:coverage}), exact on the standard-safety fragment (\cref{sec:appendixA}) and
        generative past it, predicting configurations the catalog has not named.
    \end{itemize}

    \section{Model and the Two Operators}
    \label{sec:model}

    \subsection{The Observed Causal DAG}
    \label{sec:obs-dag}

    \begin{definition}[Causal DAG]\label{def:dag}
        A \emph{causal DAG} is a directed acyclic graph $G = (M, E)$. $M$ is a set of
        \emph{messages}---atomic point events (a write, a packet, an operation, a
        letter). $E \subseteq M \times M$ is a set of \emph{dependency edges}:
        $(m', m) \in E$ means the creator of $m$ had observed $m'$ before creating $m$.
        The transitive closure of $E$ is the \emph{causal past}
        $\mathit{deps}^*(m) = \{\, m' : \text{there is a directed path } m' \rightsquigarrow m \,\}$;
        messages with no directed path between them are \emph{concurrent}.
    \end{definition}

    \begin{definition}[System]\label{def:system}
        A \emph{system} is a causal DAG $G = (M, E)$ that grows as messages are
        created, a set $N$ of \emph{observers} (nodes, processes, replicas,
        recipients), and a \emph{visibility predicate}
        $\mathrm{seen}(n, m, t) \in \{\top, \bot\}$. The \emph{visible sub-DAG} of
        observer $n$ at time $t$ is
        \[
            V(n, t) = \{\, m \in M : \mathrm{seen}(n, m, t) \,\},
        \]
        together with the edges of $G$ restricted to it. $V(n,t)$ is the whole of what
        $n$ knows.
    \end{definition}

    \begin{definition}[Monotonic Visibility, Axiom 1]\label{ax:monotonic}
        $t_1 \le t_2 \wedge \mathrm{seen}(n, m, t_1) \Rightarrow \mathrm{seen}(n, m, t_2)$.
    \end{definition}

    \begin{definition}[Creator Visibility, Axiom 2]\label{ax:creator}
        The creator of a message sees it immediately---so read-your-writes is a
        tautology \emph{for that observer}, never a guarantee. It becomes a guarantee
        only \emph{across} an observer boundary, which is what makes the boundary remark
        below load-bearing rather than bookkeeping.
    \end{definition}

    \begin{definition}[Observation Induces Dependency, Axiom 3]\label{ax:observation}
        If the creator of $m$ saw $m'$ before creating $m$, then $(m', m) \in E$:
        \[
            \mathrm{seen}(\mathrm{creator}(m), m', t_{\mathrm{create}}(m)) \Rightarrow (m', m) \in E.
        \]
    \end{definition}

    This is Lamport's happened-before relation~\cite{lamport1978} encoded as DAG
    structure---the axiom that makes the graph \emph{causal}, tying the visibility
    predicate to the edge set.

    \begin{definition}[Correct observer]\label{def:correct}
        An observer is \emph{correct} if it does not permanently crash.
    \end{definition}

    \begin{remark}[No further assumptions]\label{rem:no-delivery}
        Nothing else is assumed: no reliable channels, no bounded delays, no
        synchronized clocks, no failure model---only the DAG, the visibility predicate,
        and the three axioms. In particular, for any message an observer did not create,
        delivery is not guaranteed: $\mathrm{seen}(n, m, t)$ may remain $\bot$ for every
        $t$. This delivery-freedom is the sole hypothesis \cref{sec:filter-forces} will
        draw on.
    \end{remark}

    \begin{remark}[Seen versus retained; observer boundaries]\label{rem:seen-retained}
        $\mathrm{seen}(n, m, t)$ records that $n$ \emph{observed} $m$, not that it keeps
        $m$'s content---an observer may discard content while the fact of observation
        persists (\cref{sec:retention}'s retention bound concerns the weaker duty to keep
        dependency \emph{metadata}). And a separate process, thread, or
        connection---even co-located---is a distinct observer, since reaching it costs a
        message subject to the same $(C, O, F)$. This is precisely what gives Axiom~2 its
        content: read-your-writes is trivial \emph{within} an observer; the
        read-your-writes \emph{guarantee}---the one real systems expend effort to
        provide---is its cross-boundary form, where a second observer (a reconnected
        client, a different replica) must actually receive the write before it can read
        it. Both points are developed in \cref{sec:glossary}; here they only fix what
        $V(n,t)$ denotes.
    \end{remark}

    \subsection{Causal Closure Is a Filter: \texorpdfstring{$C(\mathit{deps})$}{C(deps)}}
    \label{sec:closure}

    Causal closure is the operator the rest of the framework rests on, and it is a
    \emph{filter}: it selects which of the DAG's dependency edges an observer is obliged
    to be closed over.

    \begin{definition}[Causal closure]\label{def:closure}
        For a \emph{dependency filter} $\mathit{deps} : M \to \mathcal{P}(M)$, the
        configuration $C(\mathit{deps})$ requires
        \[
            \forall n,m,t:\quad \mathrm{seen}(n,m,t) \;\Rightarrow\; \forall m' \in \mathit{deps}(m):\ \mathrm{seen}(n,m',t).
        \]
        An observer satisfying $C(\mathit{deps})$ never holds a message without also
        holding everything the filter attaches to it. Read through the recursion---each
        required $m'$ is itself seen, so \emph{its} filtered dependencies are seen
        too---the filter makes $V(n,t)$ downward-closed under the relation it picks out.
        Representative filters:
        \begin{itemize}
            \item $\mathit{deps}_{\mathrm{none}}(m) = \varnothing$---no filter; $V(n,t)$ need not be closed over anything.
            \item $\mathit{deps}_{\mathrm{object}}(m)$---$m$'s dependencies that share its object; closure within each object, distinct objects independent.
            \item $\mathit{deps}_{\mathrm{session}}(m)$---$m$'s dependencies in the same session; closure within a session's history.
            \item $\mathit{deps}_{\mathrm{explicit}}(m) = \{\,m':(m',m)\in E\,\}$---every immediate predecessor; full causal closure, $V(n,t)$ downward-closed under $\mathit{deps}^*$.
        \end{itemize}
    \end{definition}

    \subsection{Fork Resolution Within the Filter: \texorpdfstring{$O(\pi)$}{O(pi)}}
    \label{sec:fork}

    The causal DAG is a partial order~\cite{lamport1978}, so it leaves \emph{forks}:
    concurrent messages with no directed path between them. Concurrency is forced by
    distance---two observers far enough apart each act before the other's message
    arrives, so the causal structure itself produces the fork (\cref{fig:lightcone}). A
    fork is judged \emph{within} the sub-DAG $C$ admits: a pair the observer cannot order
    from the edges its filter obliges it to hold.

    \begin{figure}[t]
        \centering
        \includegraphics[width=0.6\linewidth]{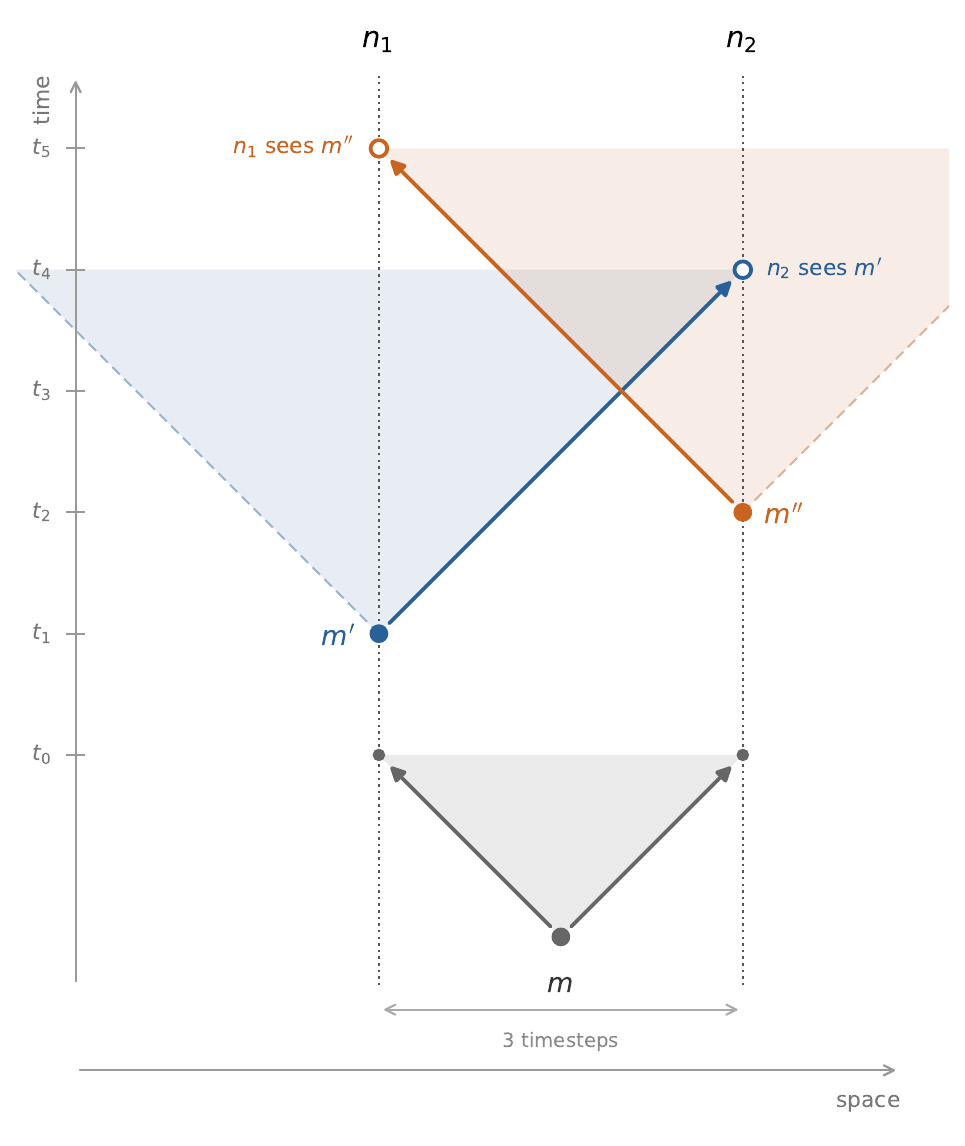}
        \caption{A fork is concurrency forced by distance. Nodes $n_1$ and $n_2$ are
        three timesteps apart, so a message takes three steps to cross and
        propagation runs along the $45^\circ$ light-cone edge. Both see the common
        ancestor $m$ at $t_0$; $n_1$ then emits $m'$ at $t_1$ and $n_2$ emits $m''$
            at $t_2$, each before the other's message can arrive. Both $m'$ and $m''$
            depend on $m$, yet neither lies in the other's past light cone: they are
            concurrent. The fork is a property of the causal structure, not a decision,
            and it is exactly what $O$ scopes.}
        \label{fig:lightcone}
    \end{figure}

    \begin{definition}[Fork resolution]\label{def:fork}
        A partition $\pi : M \to I$ sorts each message into a \emph{class}
        $i = \pi(m) \in I$; the class itself is the preimage
        $\pi^{-1}(i) \subseteq M$. The configuration $O(\pi)$ requires that within each
        class the forks of the $C$-filtered sub-DAG be resolved into a total order
        $\prec_i$ that visibility respects:
        \[
            \forall n,\; i \in I,\; m_1,m_2 \in \pi^{-1}(i),\, t:\quad \mathrm{seen}(n,m_2,t) \wedge m_1 \prec_i m_2 \;\Rightarrow\; \mathrm{seen}(n,m_1,t).
        \]
        By refinement, $O(\mathrm{trivial}) \sqsubseteq O(\pi_{\mathrm{object}}) \sqsubseteq O(\mathrm{all})$:
        \begin{itemize}
            \item $O(\mathrm{trivial})$: no class entangles two forks---the causal partial order only; concurrent messages carry no order.
            \item $O(\pi_{\mathrm{object}})$: forks entangled within an object, distinct objects disjoint.
            \item $O(\mathrm{all})$: every fork entangled into one total order over the whole DAG.
        \end{itemize}
    \end{definition}

    $O(\pi)$ asks only for \emph{a} total order over the concurrent---spacelike---messages
    in each class: the order that simultaneity withholds. How the tie is broken (delivery
    order, random, last-writer-wins, a deterministic merge) is invisible to $O$.

    \subsection{Response: \texorpdfstring{$F = (R, s)$}{F = (R, s)}}
    \label{sec:response}

    $C$, $O$, and the visibility predicate fix what an observer \emph{holds} and how it
    is \emph{ordered}; $F$ fixes what it \emph{reports}. Fix an observer $n$ and write
    its state at time $t$ as
    \[
        \sigma(n,t) = \big(\,V(n,t),\ \prec_n\,\big),
    \]
    the visible sub-DAG together with the order $\prec_n$ that $C$ and $O$ resolve over
    it (\cref{sec:closure,sec:fork}). A query $q$ arrives at time $t_0$, and $F$ answers
    in two stages---\emph{when}, then \emph{what}.

    \begin{definition}[Waiting policy]\label{def:waiting}
        $R$ is a stopping rule: watching the state trajectory
        $(\sigma(n,t))_{t \ge t_0}$ in a configured unit (wall-clock time, message count,
        version distance, or position in the $O$-order), it selects an answer time
        $t^{*} = R(q,n,t_0) \in [t_0,\infty]$.
        \begin{itemize}
            \item $R(0)$: $t^{*}$ at unit-zero---no deferral in the unit (immediate under a time unit; the whole ordered prefix handled first under the $O$-order unit).
            \item $R(\delta)$: $t^{*}$ within a bounded $\delta$ of $t_0$.
            \item $R(\infty)$: $t^{*}$ finite but unbounded---defers until a waiting condition (an acknowledgment, a quorum, a clock reading) holds, then answers.
            \item $R(\mathrm{absent})$: no unit-bound waiting---best-effort; the answer fires on a trigger or other signal rather than at an elapsed-unit threshold, and nothing is guaranteed to arrive.
        \end{itemize}
    \end{definition}

    \begin{definition}[Selection function]\label{def:selection}
        $s$ maps the query and the state at the answer time to the set of values the
        observer may report,
        \[
            s\big(q,\ \sigma(n,t^{*})\big) \in \mathcal{P}(\mathrm{Values}),
        \]
        a singleton for a single-valued read, a larger set otherwise. Six
        representatives, ranked below: $F_{\mathrm{latest}}$ (the $\prec_n$-greatest
        write), $F_{k\text{-}\mathrm{latest}}$ (any of the $k$ most recent),
        $F_{\mathrm{any\text{-}concurrent}}$ (any write concurrent with or before the
        read), $F_{\mathrm{anything}}$ (any value, even one never written),
        $F_{\mathrm{computed}}$ (a deterministic function of the state---a merge, sum, or
        maximum), $F_{\mathrm{multi}}$ (the whole concurrent set, for the caller to
        resolve).
    \end{definition}

    \begin{definition}[Response function]\label{def:rvf}
        $F$ is their composition: presented with $q$ at $t_0$,
        \[
            F(q,n,t_0) = s\big(q,\ \sigma(n,t^{*})\big), \qquad t^{*} = R(q,n,t_0).
        \]
        $R$ fixes the time at which the state is read; $s$ fixes what is read from it. We
        write $F = (R,s)$ for this composition. Its codomain is
        $\mathcal{P}(\mathrm{Values})$ throughout.
    \end{definition}

    \begin{theorem}[Restrictiveness order]\label{thm:rvf-order}
        Under inclusion of reportable-value sets,
        $F_{\mathrm{latest}} \le F_{k\text{-}\mathrm{latest}} \le F_{\mathrm{any\text{-}concurrent}} \le F_{\mathrm{anything}}$
        and, independently,
        $F_{\mathrm{latest}} \le F_{\mathrm{computed}} \le F_{\mathrm{multi}}$;
        $F_{\mathrm{computed}}$ and $F_{\mathrm{any\text{-}concurrent}}$ are incomparable.
    \end{theorem}

    \begin{proof}
        Inclusion of the reportable sets at each input; $F_{\mathrm{computed}}$ may
        report a value no write produced and $F_{\mathrm{any\text{-}concurrent}}$ one the
        deterministic function would not, so neither set contains the other.
    \end{proof}

    \begin{theorem}[Orthogonality]\label{thm:orthogonality}
        Whether a given $(C,O,R)$ skeleton is achievable depends only on $(C,O,R)$, not
        on $s$.
    \end{theorem}

    \begin{proof}
        By Axiom~3 (\cref{ax:observation}) the edges of $E$ are induced by what an
        observer saw before creating a message, never by what $F$ reported; $s$ only
        chooses a value from the state at $t^{*}$, adding no edge and changing no
        visibility, so it cannot affect which $(C,O,R)$ states are reachable.
    \end{proof}

    \subsection{What the Filter Forces}
    \label{sec:filter-forces}

    \begin{remark}[Ordering presupposes the filter]\label{rem:order-presupposes}
        Because $O$ resolves forks of the sub-DAG $C$ admits, $O$'s order presupposes
        $C$: an order respecting causal arbitration is causally closed by construction.
        The converse fails---closure fixes which edges are held, not how concurrent
        writes are ordered. This is definitional, not a theorem: $O$'s domain is $C$'s
        output.
    \end{remark}

    \begin{corollary}[Incompleteness]\label{cor:incompleteness}
        Fix a single observer enforcing $C(\mathit{deps} \neq \mathrm{none})$. By the
        no-delivery remark of \cref{rem:no-delivery}, there are histories in which a
        vertex the filter requires is absent from $V(n,t)$ for every $t$. One missing
        vertex has two faces. \emph{Closure:} $V(n,t)$ is not downward-closed, so $C$ is
        unsatisfiable, and no waiting bound repairs it. \emph{Ordering:} any resolution
        that needed the missing edge is unsound---the observer cannot tell ``concurrent''
        from ``causally ordered through a vertex it has not seen.'' $F$ can then only
        defer ($R$), report from the incomplete sub-DAG ($s$), or return $\bot$. This is
        not an impossibility theorem; it is what the definitions say happens when the
        filter cannot be met.
    \end{corollary}

    \section{Retention: the Cost of Resolving Forks}
    \label{sec:retention}

    The ordering face of \cref{cor:incompleteness} arose from a vertex that was never
    delivered. The same gap arises by \emph{choice}: an observer that has \emph{seen} the
    connecting vertex and later discarded it stands in the identical position. That is the
    retention bound, and unlike incompleteness it is a cost a system pays on itself.

    \begin{theorem}[Retention bound]\label{thm:retention}
        Take a system with non-trivial scope $O(\pi)$ ($\pi \neq \mathrm{trivial}$) and
        $C(\mathit{deps} \neq \mathrm{none})$ that \emph{holds a fork open}---defers
        resolution, carrying an unresolved pair into its state rather than ordering before
        it answers (any $R \neq R(0)$). Suppose that, while a pair $\{m_1, m_2\}$ awaits
        resolution, the observer discards the dependency metadata of a seen message $m_0$
        that lies \emph{between} them---a vertex through which the only causal paths
        joining $m_1$ and $m_2$ run. Then no algorithm reading the remaining state can
        guarantee a resolution faithful to the causal order of $m_1$ and $m_2$.
    \end{theorem}

    \begin{proof}
        The observer has seen $m_1, m_0, m_2$ and must report an order on
        $\{m_1, m_2\}$. Because every causal path between them runs through $m_0$, two
        ground-truth histories are consistent with what it has seen:
        \[
            (G)\quad m_1 \to m_0 \to m_2, \qquad\qquad (G')\quad m_2 \to m_0 \to m_1 .
        \]
        In $G$ the chain forces $m_1 \to m_2$, so a faithful order places
        $m_1 \prec m_2$; in $G'$ it forces $m_2 \prec m_1$. The two histories differ
        \emph{only} in the edges incident to $m_0$. Discarding $m_0$'s dependency metadata
        removes exactly those edges---every edge with $m_0$ as an endpoint, wherever
        stored (they existed by Axiom~3, \cref{ax:observation})---so the observer retains,
        by Axiom~1 (\cref{ax:monotonic}), only that it saw $m_1, m_0, m_2$, with no edge
        among them. Its state is therefore identical in $G$ and $G'$. A deterministic
        algorithm returns one fixed order, unfaithful in whichever history disagrees; a
        randomized one is faithful with probability below one in at least one of them. The
        pair the observer believed it was free to order was, in one world, a causal chain
        it can no longer see---the misclassification of \cref{cor:incompleteness}, here
        reached by \emph{discarding} rather than never seeing: the distinguishing
        information was retained and released, not absent from the start.
    \end{proof}

    \begin{remark}[Reading the bound]\label{rem:retention-context}
        \cref{thm:retention} is folklore in the database
        community~\cite{bottcher2019,lee2016}; the contribution is reading it as a
        constraint on \emph{all} message-passing systems, not a storage detail. It is the
        cost of reunification: a fork must be resolved only if the diverged sides are meant
        to rejoin, and then each side must have retained enough causal history to do
        so---the scope-bounded pattern of Helland~\cite{helland2007} and Brewer's
        partition-mode recovery~\cite{brewer2012}, and the window after which Burgess and
        Gerlits reclaim ``Many Worlds'' branches~\cite{burgess2022}; the partition tradeoff
        it forces is \cref{cor:partition}.
    \end{remark}

    \section{The Readability Order, \texorpdfstring{$\kappa$}{kappa}-Shortcuts, and the Factoring Dichotomy}
    \label{sec:readability}

    The configuration space has so far been a set of points---the $(C,O)$ configurations,
    with no relation among them. We now equip it with an order: when is data written under
    one configuration readable under another? The order exposes the central object of the paper, the $\kappa$ map: an
    order that refines causality supplies closure the writer's filter never demanded, and
    the $12 = 4 \times 3$ configurations $(C,O)$ collapse to $8$ classes under $\preceq$.
    This is the \emph{positive} pole of $\kappa$; \cref{sec:ratchet-scars} is the negative
    one.

    \subsection{Admissible cuts and the readability order}

    \begin{definition}[Causal order and resolved order]\label{def:two-orders}
        Two orders act on $M$. The \emph{causal order} is the happened-before partial
        order the DAG carries (\cref{ax:observation}): $m' \in \mathit{deps}^*(m)$ iff $m'$
        causally precedes $m$. The \emph{resolved order} $\prec$ is the total order the
        system lays over the forks $O$ leaves open (\cref{def:fork})---the application's
        own serialization, fixed by a tiebreak invisible to $O$. The two are a priori
        independent: $\prec$ is whatever the application produces and carries no commitment
        to the DAG's structure.
    \end{definition}

    Their relationship is the subject of this section---whether $\prec$ refines the causal
    order, and what that means for a reader. Both operators, meanwhile, are the same kind of
    object---``if you see this, you must also see that.'' $C(\mathit{deps})$ demands
    downward closure under the filter; $O(\pi)$ demands that within each class the visible
    set is a prefix of $\prec_i$. Each carves out the cuts an observer may hold.

    \begin{definition}[Admissible cut]\label{def:adm}
        Fix $G$ and a canonical resolved order $\prec$ refining the causal order (a linear
        extension of $G$; \cref{def:fork}). A cut $V \subseteq M$ is \emph{admissible}
        under $\sigma = (C, O)$, written $V \in \mathrm{Adm}(\sigma, G)$, if it is
        downward-closed under $\mathit{deps}$ (\cref{def:closure}) and a per-class prefix
        of $\prec$ within each $\pi$-class (\cref{def:fork}).
    \end{definition}

    A single canonical $\prec$ represents the entire $O$-axis: each scope reads its
    per-class prefixes off that one global order, since the scopes partition the messages,
    not the readers. (Session partitions \emph{visibility}, so it is a $C$ filter, not an
    $O$ scope; \cref{sec:fork}.)

    \begin{definition}[Readability]\label{def:readable}
        $A = (C_A, O_A)$ is \emph{readable under} $B = (C_B, O_B)$, written
        $A \preceq B$, if every $A$-admissible cut is $B$-admissible on every causal DAG:
        \[
            A \preceq B \quad\stackrel{\mathrm{def}}{\Longleftrightarrow}\quad
            \forall G:\ \mathrm{Adm}(A, G) \subseteq \mathrm{Adm}(B, G).
        \]
    \end{definition}

    A stronger configuration admits \emph{fewer} cuts, so $\preceq$ is the
    consistency-strength order restricted to $(C, O)$. Though $\mathrm{Adm}(\sigma, G)$
    depends on the chosen $\prec$, the relation $\preceq$ does not: $A$ and $B$ are scored
    against the \emph{same} $\prec$ in each $G$, and the rule below turns only on the
    scopes.

    \begin{remark}[$R$ does not participate]\label{rem:no-R}
        Admissibility is a property of a static cut, evaluated without reference to when
        messages arrive. $R$ governs only whether and when a cut becomes available, never
        whether an available cut is correctly interpretable. Readability is therefore a
        pure safety relation on $(C, O)$---the readability counterpart of
        \cref{thm:orthogonality}: there $s$ is orthogonal to achievability, here $R$ is
        orthogonal to readability.
    \end{remark}

    \subsection{The \texorpdfstring{$\kappa$}{kappa} map}

    Closure filters are ordered by inclusion, from the least
    $\mathit{deps}_{\mathrm{none}}$ to the greatest $\mathit{deps}_{\mathrm{explicit}}$,
    with $\mathit{deps}_{\mathrm{object}}$ and $\mathit{deps}_{\mathrm{session}}$
    incomparable atoms whose join is strictly below $\mathit{deps}_{\mathrm{explicit}}$ (a
    parent sharing neither object nor session stays unconstrained).

    \begin{definition}[Closure entailed by ordering---the $\kappa$ map]\label{def:kappa}
        Define $\kappa : O \to \{\text{closure filters}\}$ by
        $\kappa(O(\mathrm{trivial})) = \mathit{deps}_{\mathrm{none}}$,
        $\kappa(O(\pi_{\mathrm{object}})) = \mathit{deps}_{\mathrm{object}}$,
        $\kappa(O(\mathrm{all})) = \mathit{deps}_{\mathrm{explicit}}$. The \emph{effective
        closure} of $A$ is $C^{\mathrm{eff}}_A = \mathit{deps}_{C_A} \sqcup \kappa(O_A)$.
    \end{definition}

    \begin{definition}[Causal arbitration]\label{ass:causal-arb}
        The system's resolved order $\prec$ refines the causal order:
        $m' \in \mathit{deps}^*(m) \Rightarrow m' \prec m$; equivalently
        $\mathit{hb} \subseteq \mathit{ar}$ in Burckhardt's terms~\cite{burckhardt2014}.
    \end{definition}

    $\kappa$ is one relation with two poles. An order is $\kappa$-\emph{clean} when it
    honors the entailment---$\prec$ refines causality---and bears a
    $\kappa$-\emph{scar} (\cref{def:scar}) where it does not. The same map's structural
    features in the lattice below---ordering supplying closure the writer's filter
    lacked---are its $\kappa$-\emph{shortcuts}. \Cref{sec:readability} works the clean
    pole; \cref{sec:ratchet-scars} the scarred one.

    \begin{lemma}[Ordering entails closure]\label{lem:o-entails-c}
        Under causal arbitration (\cref{ass:causal-arb}), every $O(\pi_X)$-admissible cut is
        $\mathit{deps}_X$-closed; in particular $O(\mathrm{all})$ forces
        $\mathit{deps}_{\mathrm{explicit}}$-closure.
    \end{lemma}

    \begin{proof}
        Let $V$ be $O(\pi_X)$-admissible, $m \in V$, and $m'$ a $\mathit{deps}_X$-required
        parent of $m$. Then $m'$ lies in $m$'s $\pi_X$-class and, being a causal ancestor,
        satisfies $m' \prec m$ by causal arbitration; the per-class prefix condition of
        $O(\pi_X)$ forces $m' \in V$.
    \end{proof}

    \subsection{The readability rule and the factoring dichotomy}

    \begin{theorem}[Readability rule]\label{thm:readability}
        Under causal arbitration, $A \preceq B$ if and only if
        \[
            O_A \sqsupseteq_O O_B \quad\text{and}\quad \mathit{deps}_{C_B} \sqsubseteq C^{\mathrm{eff}}_A,
        \]
        where $\sqsupseteq_O$ is the chain
        $O(\mathrm{trivial}) \sqsubset O(\pi_{\mathrm{object}}) \sqsubset O(\mathrm{all})$.
    \end{theorem}

    \begin{proof}
        Write $\mathrm{Adm}(A,G) = \mathrm{Adm}_C(C_A,G) \cap \mathrm{Adm}_O(O_A,G)$, the
        cuts that are $\mathit{deps}_{C_A}$-closed and an $O_A$-prefix, and likewise for
        $B$. We isolate two facts about the constituents, neither relying on
        causal arbitration.

        \textbf{(O1) $O$-monotonicity.} If $O_A \sqsupseteq_O O_B$ then
        $\mathrm{Adm}_O(O_A,G) \subseteq \mathrm{Adm}_O(O_B,G)$ for every $G$. As
        $\sqsupseteq_O$ is the chain
        $O(\mathrm{trivial}) \sqsubset O(\pi_{\mathrm{object}}) \sqsubset O(\mathrm{all})$,
        it suffices to check its two covers. For
        $O(\mathrm{all}) \sqsupseteq_O O(\pi_{\mathrm{object}})$: an
        $O(\mathrm{all})$-admissible cut is an initial segment of the single global order
        $\prec$, and the restriction of an initial segment to any object class is an
        initial segment of $\prec$ on that class---an $O(\pi_{\mathrm{object}})$-prefix.
        For $O(\pi_{\mathrm{object}}) \sqsupseteq_O O(\mathrm{trivial})$:
        $O(\mathrm{trivial})$ imposes no prefix constraint, so the inclusion is vacuous.

        \textbf{(C1) Closure is conjunctive.} A cut is $\mathit{deps}_f$- and
        $\mathit{deps}_g$-closed iff it is $\mathit{deps}_{f \sqcup g}$-closed, since
        $\mathit{deps}_{f \sqcup g}$ requires exactly the parent edges required by $f$ or
        by $g$. Consequently $\mathit{deps}_{C_B} \sqsubseteq h$ implies every
        $\mathit{deps}_h$-closed cut is $\mathit{deps}_{C_B}$-closed.

        \emph{Sufficiency.} Assume $O_A \sqsupseteq_O O_B$ and
        $\mathit{deps}_{C_B} \sqsubseteq C^{\mathrm{eff}}_A$. Fix $G$ and
        $V \in \mathrm{Adm}(A,G)$. By~(O1), $V \in \mathrm{Adm}_O(O_B,G)$, so $B$'s
        ordering constraint holds. For closure: $V$ is $\mathit{deps}_{C_A}$-closed by
        hypothesis and $O_A$-admissible, so by \cref{lem:o-entails-c}---the single appeal to
        causal arbitration---$V$ is $\kappa(O_A)$-closed; by~(C1) it is
        $C^{\mathrm{eff}}_A = \mathit{deps}_{C_A} \sqcup \kappa(O_A)$-closed, hence
        $\mathit{deps}_{C_B}$-closed. Thus $V \in \mathrm{Adm}(B,G)$, and as $G, V$ were
        arbitrary, $A \preceq B$.

        \emph{Necessity.} We prove the contrapositive: if the right-hand side fails, we
        exhibit one $G$ and a cut $V \in \mathrm{Adm}(A,G) \setminus \mathrm{Adm}(B,G)$,
        which by \cref{def:readable} refutes $A \preceq B$. A single counterexample
        suffices precisely because readability quantifies over all $G$.

        \emph{Case 1: $O_A \not\sqsupseteq_O O_B$.} Then $O_B$ enforces a prefix that $O_A$
        does not. Take $G$ with two messages $m_1 \prec m_2$ sharing a $B$-class but not an
        $A$-class, on fresh objects/sessions so $\mathit{deps}_{C_A}$ requires neither in
        the other's closure. Let $V = \{m_2\}$: it is $\mathit{deps}_{C_A}$-closed (no
        required parent) and $O_A$-admissible ($m_2$ alone is a prefix of its $A$-class),
        so $V \in \mathrm{Adm}(A,G)$; but within the shared $B$-class it holds $m_2$
        without its $\prec$-predecessor $m_1$, violating the $O_B$-prefix, so
        $V \notin \mathrm{Adm}(B,G)$.

        \emph{Case 2: $\mathit{deps}_{C_B} \not\sqsubseteq C^{\mathrm{eff}}_A$.} Then some
        dependency-edge type is required by $C_B$ but lies in neither
        $\mathit{deps}_{C_A}$ nor $\kappa(O_A)$; in particular $O_A \neq O(\mathrm{all})$,
        since $\kappa(O(\mathrm{all})) = \mathit{deps}_{\mathrm{explicit}}$ dominates every
        requirement. Take $G$ to be the single edge $(m', m)$ of that type, with $m', m$ in
        distinct $O_A$-classes---possible exactly because the edge is not a
        $\kappa(O_A)$-edge. Let $V$ be the $\subseteq$-least $A$-admissible cut containing
        $m$; since $(m', m) \notin C^{\mathrm{eff}}_A$, neither $A$'s closure nor its
        $\kappa(O_A)$-prefix draws $m'$ into $V$, so $m' \notin V$. Then
        $V \in \mathrm{Adm}(A,G)$ but is not $\mathit{deps}_{C_B}$-closed (the required
        parent $m'$ is absent), so $V \notin \mathrm{Adm}(B,G)$.

        In each case $A \not\preceq B$, completing necessity.
    \end{proof}

    \begin{theorem}[Factoring dichotomy]\label{thm:dichotomy}
        Under causal arbitration, $\preceq$ does \emph{not} factor
        into $\sqsupseteq_C \times \sqsupseteq_O$; the gap is exactly the
        $\kappa$-shortcuts of \cref{lem:o-entails-c}. Without it, $\preceq$ factors:
        $A \preceq B \iff \mathit{deps}_{C_A} \sqsupseteq \mathit{deps}_{C_B}$ and
        $O_A \sqsupseteq_O O_B$. Cross-axis coupling is the signature of causally coherent
        ordering.
    \end{theorem}

    \begin{proof}
        Under causal arbitration, \cref{thm:readability} gives
        $A \preceq B \iff O_A \sqsupseteq_O O_B \wedge \mathit{deps}_{C_B} \sqsubseteq
        \mathit{deps}_{C_A} \sqcup \kappa(O_A)$. The product order is
        $A \preceq^{\times} B \iff O_A \sqsupseteq_O O_B \wedge \mathit{deps}_{C_B}
        \sqsubseteq \mathit{deps}_{C_A}$. The two agree on the $O$ coordinate and differ
        only by replacing $\mathit{deps}_{C_A}$ with $\mathit{deps}_{C_A} \sqcup
        \kappa(O_A)$; since $x \sqsubseteq y \Rightarrow x \sqsubseteq y \sqcup z$, we have
        $\preceq^{\times} \subseteq \preceq$. The inclusion is strict whenever
        $\kappa(O_A)$ supplies a closure edge $\mathit{deps}_{C_A}$ lacks: with
        $A = (\mathrm{none}, \mathrm{all})$ and $B = (\mathrm{explicit}, \mathrm{trivial})$,
        $C^{\mathrm{eff}}_A = \kappa(O(\mathrm{all})) = \mathit{deps}_{\mathrm{explicit}}
        \sqsupseteq \mathit{deps}_{C_B}$ gives $A \preceq B$, while
        $\mathit{deps}_{C_A} = \mathit{deps}_{\mathrm{none}} \not\sqsupseteq
        \mathit{deps}_{\mathrm{explicit}}$ gives $A \not\preceq^{\times} B$. The gap is
        exactly the set of $\kappa$-edges (\cref{lem:o-entails-c,cor:nonfactor}).

        Drop causal arbitration. Then \cref{lem:o-entails-c} fails: the resolved order may
        place a causal parent $m'$ after its child, $m \prec m'$ though
        $m' \in \mathit{deps}^*(m)$, so an
        $O(\pi_X)$-prefix cut need not be $\mathit{deps}_X$-closed. Ordering then entails no
        closure, $\kappa \equiv \mathit{deps}_{\mathrm{none}}$, and $C^{\mathrm{eff}}_A =
        \mathit{deps}_{C_A}$. Facts~(O1) and~(C1) in the proof of \cref{thm:readability} did
        not use causal arbitration, so the same sufficiency/necessity argument with this
        $\kappa$ yields $A \preceq B \iff O_A \sqsupseteq_O O_B \wedge \mathit{deps}_{C_A}
        \sqsupseteq \mathit{deps}_{C_B}$---the product order (\cref{fig:readability-lww}).
        Hence the cross-axis coupling is present under causal arbitration and absent
        without it.
    \end{proof}

    \begin{remark}[The dichotomy is a cliff, not a slope]\label{rem:cliff}
        Causal arbitration enters all-or-nothing. Because $\preceq$ quantifies over
        \emph{all} $G$ (\cref{def:readable}), a $\kappa$-shortcut $A \preceq B$ survives
        only if every execution honours causal arbitration on the relevant edges: one $G$
        that inverts a single such edge already carries an $A$-admissible cut omitting a
        $\mathit{deps}_{C_B}$-required parent (the Case~2 witness of
        \cref{thm:readability}), refuting $A \preceq B$. So any positive rate of causal
        inversion collapses every $\kappa$-edge at once, returning the pure product order;
        there is no intermediate regime. The smooth object, for anyone wanting graceful
        degradation, is an \emph{$\varepsilon$-readability} relation relaxing $\forall G$ to
        ``all but an $\varepsilon$-fraction.'' The cliff is a property of the universal
        quantifier, not of the systems.
    \end{remark}

    \begin{figure}[t]
        \centering
        \includegraphics[width=\linewidth]{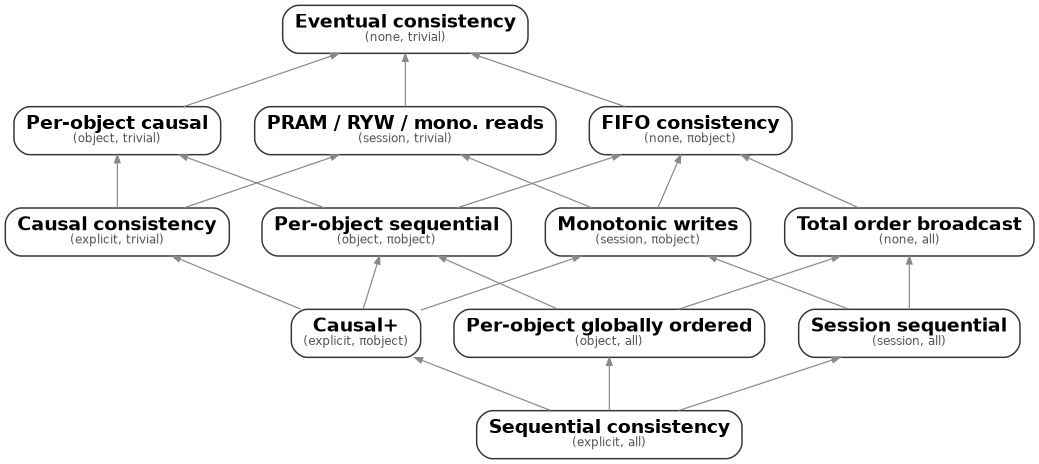}
        \caption{Readability when causal arbitration fails. With $\prec$ free to invert a
        causal edge, $\kappa$ is inert: no cross-axis shortcut holds, all $12$
            configurations are distinct, and $\preceq$ is the product order
            $\sqsupseteq_C \times \sqsupseteq_O$---the degenerate counterpart of the
            causal-arbitration lattice, and the second half of \cref{thm:dichotomy}. (An
            edge $A \to B$ means data written under $A$ is readable under $B$; the figure is
            rendered with physical-clock last-writer-wins as the causal violation.)}
        \label{fig:readability-lww}
    \end{figure}

    \begin{corollary}[Readability does not factor]\label{cor:nonfactor}
        $\preceq$ is strictly coarser than the product order
        $\sqsupseteq_C \times \sqsupseteq_O$: there exist $A \preceq B$ with
        $\mathit{deps}_{C_A} \not\sqsupseteq \mathit{deps}_{C_B}$. Every such pair arises
        through $\kappa$---the writer's ordering supplying closure its $C$ did not.
    \end{corollary}

    The coupling is one-directional. In this safety relation only the $O \to C$ edge
    appears: ordering supplies closure (\cref{lem:o-entails-c}), and it survives the
    removal of $R$ (\cref{rem:no-R}). Closure never supplies ordering, and the failure-side
    coupling of \cref{cor:incompleteness}---which lives in waiting, not in a static
    cut---does not appear here at all.

    \begin{corollary}[Closure is unobservable under a global order]\label{cor:c-blind}
        Under $O(\mathrm{all})$ the closure scope is read-irrelevant:
        $(\mathrm{none},\mathrm{all})$, $(\mathrm{object},\mathrm{all})$,
        $(\mathrm{session},\mathrm{all})$, and $(\mathrm{explicit},\mathrm{all})$ form a
        single $\preceq$-equivalence class---the strongest class, readable under every
        configuration.
    \end{corollary}

    The collapse is countable. The $12 = 4 \times 3$ configurations reduce to $8$ classes
    under $\preceq$: at $O(\mathrm{trivial})$ ($\kappa = \mathit{deps}_{\mathrm{none}}$) the four $C$-scopes stay
    distinct; at $O(\pi_{\mathrm{object}})$, $\kappa = \mathit{deps}_{\mathrm{object}}$
    absorbs $\mathrm{none}$ into $\mathrm{object}$, leaving three; at $O(\mathrm{all})$ they
    collapse to one (\cref{cor:c-blind}). $4 + 3 + 1 = 8$. The eight are the
    $\kappa$-shortcuts made concrete---each is an order supplying closure its
    filter did not (\cref{fig:readability}).

    \begin{figure}[t]
        \centering
        \includegraphics[width=\linewidth]{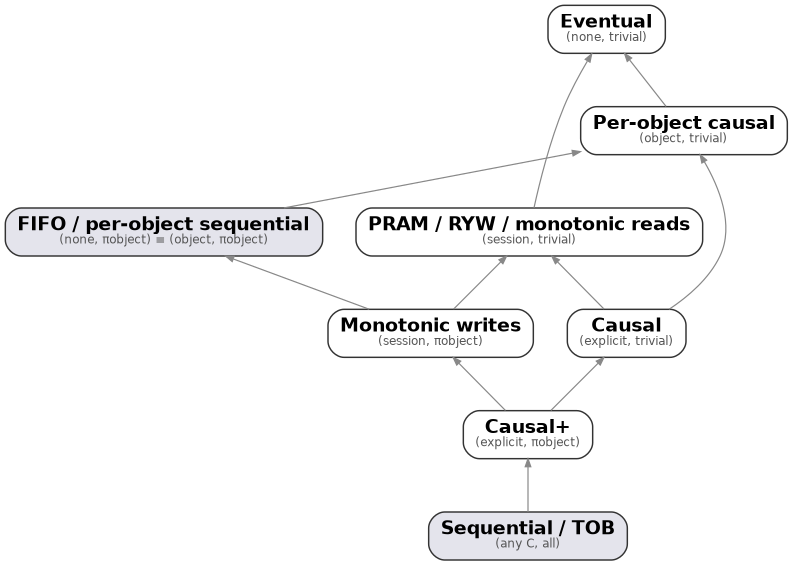}
        \caption{The readability order $\preceq$ as its $8$ equivalence classes. Each node
        is one class, and its subtitle gives the $(C,O)$ configurations the class
        contains: six are singletons, and the two shaded classes are the nontrivial
            $\kappa$-collapses---$(\mathrm{none},\pi_{\mathrm{object}}) \equiv
            (\mathrm{object},\pi_{\mathrm{object}})$, and all four $C$-scopes at
            $O(\mathrm{all})$---so the $12 = 4 \times 3$ configurations reduce to $8$. An
            edge $A \to B$ means data written under $A$ is readable under $B$ (a cover; $A$
            strictly stronger). Strongest (the $O(\mathrm{all})$ class) at the bottom;
            weakest $(\mathrm{none},\mathrm{trivial})$ at the top.}
        \label{fig:readability}
    \end{figure}

    \section{The Mergeability Relation}
    \label{sec:merge}

    Readability (\cref{sec:readability}) was asymmetric---source to target. Its symmetric
    counterpart asks when two divergent views \emph{reconcile}. Take two cuts $V_A, V_B$
    over the same DAG, each resolved under its own order $\prec_A, \prec_B$---two replicas,
    two frames---and ask whether the union $V_A \cup V_B$ is admissible under a target
    $M = (C_M, O_M)$. The configuration's two axes behave oppositely. \emph{Closure merges
    for free:} the union of two $\mathit{deps}$-closed sets is $\mathit{deps}$-closed, so
    $C$ never obstructs. \emph{Ordering is the whole story:} $V_A \cup V_B$ is
    $O_M$-admissible exactly when the two source orders, restricted to each $O_M$-class,
    are jointly acyclic---when they never disagree on a co-ordered pair.

    \begin{definition}[Merge]\label{def:merge}
        Cuts $V_A, V_B$ over $G$, resolved under $\prec_A, \prec_B$, \emph{merge} under
        $M = (C_M, O_M)$---written $\mathrm{merge}(V_A, V_B; M)$---if $V_A \cup V_B$ is
        $M$-admissible (\cref{def:adm}) under some order extending $\prec_A$ and $\prec_B$
        within each $O_M$-class.
    \end{definition}

    \begin{theorem}[Mergeability]\label{thm:merge}
        $\mathrm{merge}(V_A, V_B; M)$ holds if and only if
        \begin{enumerate}
            \item \textup{(closure---free)} $V_A$ and $V_B$ are each $C_M$-closed; and
            \item \textup{(order---binding)} within every $O_M$-class, $\prec_A$ and
            $\prec_B$ agree on the pairs they both order; equivalently,
            ${\prec_A} \cup {\prec_B}$ restricted to each $O_M$-class is acyclic.
        \end{enumerate}
    \end{theorem}

    \begin{proof}
        \emph{Closure.} If $V_A, V_B$ are $C_M$-closed, so is $V_A \cup V_B$: any
        $x \in V_A \cup V_B$ lies in some source, which already contains every parent
        $\mathit{deps}_{C_M}(x)$ demands, so those parents lie in the union. Closure is
        monotone under union, and~(1)---the readability $C$-condition
        (\cref{thm:readability}) on each source---is the only closure obligation, automatic
        for causally-closed prefixes. Conversely an $M$-admissible union is $C_M$-closed,
        and a source omitting a required parent the other does not supply is already
        inadmissible.

        \emph{Order.} $V_A \cup V_B$ is $O_M$-admissible iff within each $O_M$-class some
        total order extends both sources and keeps the visible set a prefix; such an order
        exists iff ${\prec_A} \cup {\prec_B}$ restricted to the class is acyclic. A cycle
        collapses to a contradicted pair: each source totally orders its class, so any
        cyclic walk contains a two-edge subpath $m_1 \prec_X m_2 \prec_X m_3$ within one
        source $X$, which forces $m_1 \prec_X m_3$ by transitivity and shortens the cycle;
        iterating leaves a $2$-cycle $m_1 \prec_A m_2$, $m_2 \prec_B m_1$ on a
        jointly-ordered pair. Acyclicity is thus exactly the pairwise agreement of~(2); a
        topological sort then witnesses the merged order, and a contradicted pair is the
        conflict.
    \end{proof}

    \begin{corollary}[Witness and conflict]\label{cor:merge-witness}
        When \cref{thm:merge} holds, a merged order $\prec_M$ is any per-class topological
        sort of ${\prec_A} \cup {\prec_B}$; when it fails, a minimal cycle---a single
        oppositely-ordered pair---names the two events that cannot occupy one frame. One
        acyclicity test decides both.
    \end{corollary}

    \begin{remark}[Cheap check]\label{rem:merge-cheap}
        Because each $O_M$-class is totally ordered within each source (\cref{def:fork}),
        the cycle check collapses to \emph{pairwise agreement on shared pairs}: any cycle
        would put a $2$-edge subpath inside one total order, forcing the closing edge and
        contradicting the other source on a jointly-touched pair. Local check, no global
        traversal.
    \end{remark}

    \begin{remark}[Closure merges; ordering breaks]\label{rem:merge-asymmetry}
        Readability is entangled across both axes (the $\kappa$-shortcut: $O$ supplies
        closure $C$ did not demand). Mergeability is \emph{not}: $C$ joins monotonically (a union of
        $C$-closed views is $C$-closed), so all merge fragility lives on $O$.
        \emph{Closure is a join-semilattice; order is fragile under join.}
    \end{remark}

    \begin{theorem}[The mergeable frontier]\label{thm:merge-frontier}
        Mergeability is $C$-invariant and falls monotonically as $O_M$ strengthens:
        \begin{itemize}
            \item $O_M = O(\mathrm{trivial})$: clause~(2) is vacuous---no class entangles
            two forks---so any two $C_M$-closed views merge.
            \item $O_M = O(\pi_{\mathrm{object}})$: mergeable iff $\prec_A, \prec_B$ agree
            within each object; cross-object disagreement is harmless.
            \item $O_M = O(\mathrm{all})$: mergeable iff the two frames agree
            globally---iff, on their shared events, they are the same frame.
        \end{itemize}
        Above $O(\mathrm{trivial})$, agreement is guaranteed only by a \emph{confluent}
        resolution---a deterministic, observer-independent tiebreak that makes
        $\prec_A = \prec_B$ on shared classes by construction; under any other tiebreak it
        holds only by chance. The always-mergeable region is therefore
        $O(\mathrm{trivial})$ together with the confluent resolutions at any $O$.
    \end{theorem}

    \begin{proof}
        Every clause reads off \cref{thm:merge}, whose clause~(1) is $C$-free and whose
        clause~(2) is the only binding condition.

        \emph{$C$-invariance.} A causally-closed prefix is $\mathit{deps}_{C_M}$-closed for
        every $C_M$, since $\mathit{deps}_{C_M} \sqsubseteq \mathit{deps}_{\mathrm{explicit}}$
        demands no parent causal closure does not. So clause~(1) holds at every $C_M$ and
        drops out; merge turns on clause~(2), which never mentions $C_M$.

        \emph{Monotonicity.} If $O_M \sqsubseteq O_M'$ then each $O_M'$-class is a union of
        $O_M$-classes, so the pairs clause~(2) constrains at $O_M$ are a subset of those it
        constrains at $O_M'$. Agreement at the stronger target therefore implies agreement
        at the weaker: $\mathrm{merge}(V_A,V_B;M') \Rightarrow \mathrm{merge}(V_A,V_B;M)$.

        \emph{The three rungs.} At $O(\mathrm{trivial})$ no class entangles two forks, so
        clause~(2) ranges over no pairs and holds vacuously---any two $C_M$-closed views
        merge. At $O(\pi_{\mathrm{object}})$ the classes are the per-object fork sets, so
        clause~(2) is agreement within each object and is blind to cross-object pairs. At
        $O(\mathrm{all})$ the single class is the whole shared event set, so clause~(2) is
        agreement on every shared pair: $\prec_A$ and $\prec_B$ are one order on what both
        observers hold.

        \emph{The confluent frontier.} Fix $O_M \sqsupset O(\mathrm{trivial})$. If the
        resolution is confluent---each class ordered by a function of its messages
        alone---then two observers holding the same class order it identically, so
        $\prec_A = \prec_B$ there and clause~(2) holds for every divergence: always
        mergeable. If it is not confluent, some class admits two observation histories that
        order a pair oppositely; the views realizing those histories jointly populate the
        class and violate clause~(2), so the configuration is not always mergeable. Hence
        above $O(\mathrm{trivial})$ a configuration is always mergeable iff its resolution
        is confluent, and with $O(\mathrm{trivial})$---always mergeable
        unconditionally---this is the stated region.
    \end{proof}

    This is the formal version of a familiar operational fact: two divergent
    last-writer-wins replicas anti-entropy into an eventually-consistent or per-object
    state, but never into a single global order. (Linearizability demands more
    still---real-time precedence---and is taken up in \cref{sec:boundary}.) It is the
    relativity reading one step further: divergent simultaneity across frames reconciles
    freely as long as no observer demands global simultaneity. $O(\mathrm{all})$ is that
    demand, and it is the one that cannot be met after the fact.

    The incompleteness \emph{structure} (\cref{cor:incompleteness}) and this merge
    frontier both fall out of the same $(C,O)$ filter/resolution shape.

    \subsection{Corollaries of the Merge Rule}
    \label{sec:merge-corollaries}

    \begin{corollary}[Partition forces availability-or-ordering]\label{cor:partition}
        A partition yields divergent frames; under $O > \mathrm{trivial}$ they may be
        unmergeable; a protocol that meets an unmergeable divergence must either decline
        to answer (lose availability) or weaken $O$ (relax ordering).
    \end{corollary}

    \begin{remark}[Relativity correspondence---safe version only]\label{rem:relativity}
        A $\kappa$-clean global $O(\mathrm{all})$ order exists iff
        $2\varepsilon \le d_{\min}$ (\cref{thm:clockbound}); below threshold only
        per-frame orders survive---the merge r\'egime. This ties the clock bound, the
        merge relation, and the title into one line. The spacetime / global-time-function
        reading is explicitly \emph{interpretation}.
    \end{remark}

    \section{The Consistency Ratchet and \texorpdfstring{$\kappa$}{kappa}-Scars}
    \label{sec:ratchet-scars}

    \Cref{sec:readability} worked the clean pole of the $\kappa$ map: an order that
    refines causality is readable downward, and the structural slack it leaves---ordering
    supplying closure the writer's filter never demanded---is a $\kappa$-shortcut. This
    section works the scarred pole. Two results follow. The \emph{ratchet}: a consistency
    level lost on a datum under reading-and-rewriting is never regained. And \emph{Detection
        $=$ Prevention}: a system that took the $\kappa$-shortcut, or never refined causality at
    all, leaves a permanent scar the moment its order contradicts the causal one---and
    detecting that scar costs exactly the causal record the shortcut let it discard.

    \subsection{The Consistency Ratchet}
    \label{sec:ratchet}

    Readability (\cref{def:readable}) is a statement about \emph{reading}: data written
    under $A$ can be read under any $B$ with $A \preceq B$. The dynamic consequence is about
    \emph{reading and writing back}, and it runs the other way: a value read at a weaker
    level and re-written can never be recovered at its original level. The readability order
    is a one-way ratchet under read--write.

    \paragraph*{Value migration.} Distinguish two write-backs. A \emph{version append}
    writes $m_1$ with a dependency edge to the existing $m_0$; by Axiom~3 and transitivity
    $m_1$'s causal past then includes all of $m_0$'s, and nothing is lost. A \emph{value
    migration} instead reads the value of $m_0$ and writes it as a fresh datum---into a new
    store, a new key, a cache, an export---carrying provenance only for what the reader
    observed. This is what migration and ETL tooling do, and it is where the ratchet bites.
    A level-$B$ reader is obligated to observe only the minimal $B$-admissible cut
    $V_B(m_0)$: $m_0$ together with exactly the closure $B$ forces. By Axiom~3 a write may
    depend only on what its creator observed, so the migrated datum can attach provenance for
    $V_B(m_0)$ and no more.

    \begin{theorem}[Consistency ratchet]\label{thm:ratchet}
        Let $m_0$ have native level $A$ and be value-migrated by a level-$B$ reader to a
        datum $m_1$. Then:
        \begin{enumerate}
            \item \emph{(Ceiling.)} $m_1$ is readable at $B$ and, in the worst case, at no
            level stronger than $B$: the strongest level whose provenance $m_1$ is guaranteed
            to carry is exactly $B$.
            \item \emph{(Loss.)} The provenance dropped is
            $\mathrm{closure}_A(m_0) \setminus \mathrm{closure}_B(m_0)$, nonempty whenever
            $A \succ B$ in $m_0$'s neighborhood.
            \item \emph{(Irreversibility.)} The drop is permanent. Edges are fixed at
            creation (Axioms~1,~3), so no later write can add the missing provenance to
            $m_1$; and once the source is decommissioned or its surplus history
            garbage-collected, the dropped messages are unrecoverable (\cref{thm:retention}).
        \end{enumerate}
    \end{theorem}

    \begin{proof}
        $m_1$'s recorded past is $V_B(m_0) = \mathrm{closure}_B(m_0)$. A reader at level $X$
        can read $m_1$ honestly iff $\mathrm{closure}_X(m_0) \subseteq V_B(m_0)$. For $X = B$
        this holds; for $X \succ B$, $\mathrm{closure}_X(m_0) \supseteq \mathrm{closure}_B(m_0)$
        with the inclusion strict exactly when $A \succ B$ forces additional parents or
        predecessors, so $X$ is not guaranteed---giving (1) and (2). For (3): adding the
        missing messages to $V_B(m_0)$ after the fact would require the migrator to have
        observed them, contradicting that it read at $B$; and reconstructing them later
        requires they were retained, which \cref{thm:retention} does not guarantee once they
        are unreferenced.
    \end{proof}

    \begin{corollary}[Composition]\label{cor:ratchet-compose}
        A pipeline of value-migrations through levels $B_1, B_2, \dots, B_k$ from a source of
        level $A$ produces output of level $A \sqcup B_1 \sqcup \cdots \sqcup B_k$---the join
        (weakest hop) in the readability order. The ratchet only turns down: re-reading a
        later stage at a stronger level cannot recover provenance an earlier stage did not
        observe.
    \end{corollary}

    \begin{proof}
        Iterate \cref{thm:ratchet}: stage $j$ caps its output at $B_j$, and a stage cannot
        observe provenance its input never carried, so the carried level after $k$ hops is the
        meet of the per-hop ceilings, i.e.\ the join $A \sqcup B_1 \sqcup \cdots \sqcup B_k$ in
        $\preceq$. Joins are order-monotone, so no later (stronger) hop lowers it back toward
        $A$.
    \end{proof}

    \Cref{thm:ratchet,cor:ratchet-compose} make these losses exact rather than empirical.
    Migrating sequential data ($C(\mathrm{explicit}), O(\mathrm{all})$) through an eventual
    reader ($C(\mathrm{none}), O(\mathrm{trivial})$) drops to the join---eventual---losing both
    order and closure; through a causal reader ($C(\mathrm{explicit}), O(\mathrm{trivial})$) it
    keeps closure but loses order, since $O(\mathrm{trivial})$ is the weaker hop. Migration
    through $(C(\mathrm{none}), O(\mathrm{all}))$ loses nothing: by \cref{cor:c-blind} that
    configuration is read-equivalent to sequential, so it is not a weaker hop at all. A
    three-stage pipeline that weakens to PRAM and then re-requests $O(\mathrm{all})$ stays at
    PRAM by \cref{cor:ratchet-compose}: the final stage cannot recover what the middle stage
    discarded.

    \paragraph*{Consequences for migration tooling.} All three are corollaries of ``the
    weakest hop sets the ceiling.'' To preserve level $A$, \emph{every} stage must read at a
    level $\preceq A$: a single weaker hop---a cache, an eventually-consistent replica read, a
    dump taken without causal metadata---caps the result permanently, even if every other
    stage is sequential. You cannot migrate \emph{up}: re-importing weak data into a strong
    store yields a strong store holding weak-provenance data, since a store's configuration is
    the ceiling for \emph{new} writes, not a property retrofitted onto migrated values. And the
    damage is silent: by Axiom~3 the migrated datum is internally well-formed at every
    level---the loss is visible only against the source, so provenance must be checked
    \emph{before} the source is retired, not after. The permanence in clause~(3) is conditional
    on that record being discarded: an $R(\mathrm{absent})$ log retains the full ordered or
    closed history un-applied---the ``unread log'' configurations of \cref{tab:coverage} (a
    write-ahead log before recovery, a replicated log before apply)---so a later reader can
    re-derive the lost level from it until the log itself is collected.

    \subsection{Consistency Scars and the Detection\texorpdfstring{\,$=$\,}{ = }Prevention Bound}
    \label{sec:scars}

    The ratchet (\cref{thm:ratchet}) says a consistency level, once lost on a datum, is not
    regained by further reads and writes. We now ask what \emph{creates} such a loss in a
    running system, and---more sharply---whether a system can \emph{detect} that it has
    happened. The answer is an impossibility, and it turns on \emph{which} record is kept:
    detecting a scar means checking the resolved order $\prec$ against the causal order
    $\mathit{deps}^*$, so it is retaining $\mathit{deps}^*$---not storing $\prec$ itself,
    however completely---that confers detection. A log holding the full resolved order can
    replay it and still cannot tell whether that order inverted causality. And retaining
    $\mathit{deps}^*$ is exactly the capability that would have prevented the scar.

    \begin{definition}[$\kappa$-scar]\label{def:scar}
        By $C$ and Axiom~3 (\cref{ax:observation}) the causal DAG is never violated---no
        observer holds $m$ without $\mathit{deps}^*(m)$---so $\mathit{deps}^*$ is intact in
        every reachable state. A scar lives one level up, in the resolved order $\prec$
        (\cref{def:two-orders}) a system imposes on top of the DAG: the serialization induced
        by a physical timestamp, a sequence counter, or a commit position. Fix a system that
        \emph{intends} causal arbitration (\cref{ass:causal-arb})---its $\prec$ is meant to
        refine the causal order. A \emph{$\kappa$-scar} is a durably stored pair $(m', m)$ with
        $m' \in \mathit{deps}^*(m)$ yet $m \prec m'$: the resolved order has placed a causal
        ancestor \emph{after} its descendant. This is a property of $\prec$ against the DAG,
        not of delivery---$m'$ still arrived before $m$; the system merely \emph{ordered} them
        backwards, a transient breach of its own invariant that the store now records
        permanently. A history is \emph{$\kappa$-clean} when $\prec$ is a linear extension of
        the causal order (\cref{ass:causal-arb}).
    \end{definition}

    A $\kappa$-scar does not heal under normal operation. The offending $\prec$-position is
    durably committed (\cref{def:scar}), and by \cref{thm:ratchet} no read--write transformation
    restores a position lost to a weaker order. A scar is removed only by re-deriving $\prec$
    from retained causal history---a migration, itself bounded by the retention bound
    (\cref{thm:retention}) and impossible once the surplus history is collected. Scars are thus
    written by any \emph{seam} at which a second arbitration authority, not causally
    synchronized with the first, touches causally-linked data: a clock excursion under
    physical-timestamp ordering (\cref{thm:clockbound}), a sequence-counter reset on failover,
    the merge of two independently-ordered partitions. This is the exposure the $\kappa$-shortcut buys: a system that leaned on $\prec$ to
    stand in for a retained causal record has nothing left to check $\prec$ against.

    \begin{theorem}[Detection $=$ Prevention, under endogenous arbitration]\label{thm:detect-prevent}
        Assume \emph{endogenous arbitration}: the system's order is a function only of the
        state it retains, taking no environmental input---in particular, no reading of a
        synchronized physical clock. The following are equivalent:
        \begin{enumerate}
            \item the system can decide, for any of its histories, whether that history is
            $\kappa$-clean (\emph{detect} scars);
            \item the system retains the causal-dependency relation $\mathit{deps}^*$ on its
            messages;
            \item the system can emit a $\kappa$-clean order (\emph{prevent} scars).
        \end{enumerate}
        In particular, no endogenous system can detect a $\kappa$-scar it could not have
        prevented.
    \end{theorem}

    \begin{remark}[Clock-driven arbitration is the boundary, not a counterexample]\label{rem:detect-clock}
        The endogenous hypothesis is load-bearing, and it excludes exactly the physical-clock
        systems. When $2\varepsilon \le d_{\min}$ (\cref{thm:clockbound}), a clock-timestamp
        order is $\kappa$-clean---it \emph{prevents} scars, (3)---while retaining no
        $\mathit{deps}^*$, so it has no (2). That is not a counterexample to the equivalence but
        a system outside its scope: it buys ordering from an \emph{environmental} clock reading
        rather than from retained state, which is why the hypothesis names retained state
        specifically (\cref{cor:auditability}).
    \end{remark}

    \begin{proof}
        $(1)\Rightarrow(2)$. Deciding $\kappa$-cleanliness requires, for each stored pair,
        knowing whether $m' \in \mathit{deps}^*(m)$ (\cref{def:scar}). Suppose the retained
        state does not determine $\mathit{deps}^*$. Then there are two histories $G_1, G_2$
        with identical retained state but differing on some pair: $m' \in \mathit{deps}^*(m)$
        in $G_1$, $m' \notin \mathit{deps}^*(m)$ in $G_2$. A single stored order with
        $m \prec m'$ is a scar in $G_1$ and clean in $G_2$, so no procedure on the retained
        state is correct for both. Hence detection requires retaining $\mathit{deps}^*$. (This
        is the indistinguishability argument of \cref{thm:retention} applied to detection
        rather than resolution.)

        $(2)\Rightarrow(3)$. Given $\mathit{deps}^*$, the arbitration may emit any linear
        extension of $\mathit{deps}^*$ (a topological order), which is $\kappa$-clean by
        construction.

        $(3)\Rightarrow(2)$. To place each message after all of its causal ancestors, the
        arbitration must, at the moment it places $m$, know $\mathit{deps}^*(m)$. Absent an
        environmental input---the endogenous hypothesis---the only source of that knowledge is
        retained state: a clock reading would supply a $\kappa$-clean placement without
        $\mathit{deps}^*$ (the boundary case of \cref{rem:detect-clock}), but the hypothesis
        excludes it, so knowing $\mathit{deps}^*(m)$ requires retaining it. Placements are
        spread across the whole run, and any later descendant of $m$ must in turn be placed
        after $m$ and its ancestors; the ancestor sets therefore cannot be discarded after a
        single placement---they must stay available as each new message arrives. Maintaining
        $\mathit{deps}^*(m)$ for every $m$ throughout the run is exactly retaining the relation
        $\mathit{deps}^*$.

        $(2)\Rightarrow(1)$. Given $\mathit{deps}^*$, scanning the stored order for a pair with
        $m \prec m'$ and $m' \in \mathit{deps}^*(m)$ decides cleanliness.
    \end{proof}

    \begin{corollary}[Auditability inherits the metadata bound]\label{cor:auditability}
        Retaining $\mathit{deps}^*$ on multi-valued registers requires per-message metadata of
        size $\Omega(\min\{n, s\} \cdot \lg k)$ (\cref{cor:retention-quant}, after~\cite{attiya2015}).
        Hence a system that achieves causal arbitration by synchronized physical clocks---meeting
        the bounded-skew condition of \cref{thm:clockbound} rather than retaining
        $\mathit{deps}^*$---cannot detect its own $\kappa$-scars: it orders causally while its
        clock bound holds, but retains no evidence to audit afterward. Of the two ways to
        arbitrate causally---retained metadata, or a clock tight enough to preserve
        $\mathit{hb}$ (\cref{sec:clocks})---only the metadata is auditable; the clock orders but
        leaves nothing to check against.
    \end{corollary}

    \begin{remark}[The downgrade species needs a check at the claimed level]\label{rem:downgrade}
        \Cref{def:scar} is the \emph{wrong-order} species: a causal inversion, which a
        $\mathit{deps}^*$-equipped monitor catches. A transient downgrade from $O(\mathrm{all})$
        to a weaker $O$ leaves the causal order intact---it is $\kappa$-clean---yet fails to
        make the total-order decision $O(\mathrm{all})$ requires. Two observers may then commit
        dependent writes premised on contradictory orderings of one concurrent pair
        $\{a, b\}$---one history forces $a$ before $b$, the other $b$ before $a$. The committed
        orders, unioned over the disputed pairs, contain a cycle ($a \to b \to a$) that no
        single total order can extend, so no global serialization exists and none ever will (the
        commits are durable): a \emph{frozen order cycle}. This \emph{missing-order} scar is
        invisible to $\kappa$-detection (the causal check passes) and is exposed only by
        verifying the claimed level directly. Verifying sequential consistency of a history is
        NP-complete~\cite{gibbonskorach1997}; continuous online detection of missing-order
        scars at the sequential level is therefore intractable in the worst case. The cheap
        monitor cannot see them and the complete monitor cannot run continuously---a
        detectability gap with no point that is both fast and sound.
    \end{remark}

    \begin{remark}[Scope]\label{rem:detect-scope}
        \Cref{thm:detect-prevent} concerns \emph{certain} detection of \emph{all}
        $\kappa$-scars. Probabilistic monitors that sample a subset of dependency edges can
        flag some scars with sub-$\Omega(\min\{n,s\}\lg k)$ state, trading coverage for misses;
        the equivalence bounds only complete detection.
    \end{remark}

    \section{The Single-Observer Boundary: Linearizability Is Not One System}
    \label{sec:boundary}

    LCC characterizes \emph{single-observer} consistency: each observer reads its order off
    its own DAG---the causal order from $C$, the in-filter resolution from $O$
    (\cref{sec:fork}). Linearizability is the standard case that escapes this frame. It is a
    single global total order that additionally respects \emph{real-time} precedence:
    Burckhardt's \emph{returns-before} relation $\mathit{rb}$, where $e_1\,\mathit{rb}\,e_2$
    when $e_1$'s response precedes $e_2$'s invocation in real time~\cite{burckhardt2014}. We
    show it is irreducibly a \emph{composite} of two message-passing systems, for two
    independent reasons: the global order cannot be assembled from local order-scopes
    (\cref{thm:globality}), and the real-time relation is not a function of any one observer's
    DAG (\cref{cor:linearizability}).

    \begin{theorem}[Order-scope globality]\label{thm:globality}
        Let order-scopes $\pi_1, \dots, \pi_m$ act on a message set $M$, and let the
        \emph{co-class graph} $\Gamma$ have vertex set $M$ with an edge $\{a,b\}$ whenever
        $\pi_j(a) = \pi_j(b)$ for some $j$. The scopes realize a total order on $M$ if and only
        if that order is a Hamiltonian path in $\Gamma$ along which every edge agrees with the
        within-class order of the scope that induces it. Consequently:
        \begin{enumerate}
            \item $O(\mathrm{all})$ totally orders every $M$: its co-class graph is complete.
            \item Connectivity is necessary: if the scopes totally order $M$, the join
            $\pi_1 \vee \cdots \vee \pi_m$ is a single class on $M$.
            \item If the scopes are proper and \emph{cross-cutting}---some pair $a,b$ has
            $\pi_j(a) \neq \pi_j(b)$ for every $j$---they leave $\{a,b\}$ unordered. A global
            order cannot be assembled from cross-cutting proper scopes.
        \end{enumerate}
    \end{theorem}

    \begin{proof}
    ($\Rightarrow$)
        Suppose the scopes realize a total order $<$ on $M$. Each $O(\pi_j)$
        constrains order only within its classes (\cref{def:fork}), so a pair is ordered only
        along a co-class edge or transitively through a chain of them. Take elements $v_i,
        v_{i+1}$ consecutive in $<$. If no scope co-classes them there is no edge ordering the
        pair directly; and transitivity cannot help, since consecutive elements have nothing
        strictly between them to route through. The pair would then be unordered, contradicting
        totality. So every consecutive pair is a $\Gamma$-edge whose inducing scope orders it in
        agreement with $<$: the order is a consistent Hamiltonian path. ($\Leftarrow$) Along a
        consistent Hamiltonian path each step is ordered by its inducing scope, and transitivity
        closes these into the full chain.

        (1) Under $O(\mathrm{all})$ the one class carries a single total order; its chain is a
        Hamiltonian path on the complete co-class graph. (2) A Hamiltonian path is connected,
        and the components of $\Gamma$ are exactly the classes of the join
        $\pi_1 \vee \cdots \vee \pi_m$, so a connected $\Gamma$ is a single join-class. (3) A
        cross-cutting pair $a,b$ shares no class under any $\pi_j$, so $\{a,b\}$ carries no
        $\Gamma$-edge and no scope orders it.
    \end{proof}

    The cross-cutting hypothesis of clause~(3) is the realistic one: object, session, and
    observer are independent labels, so a message to a different object \emph{and} a different
    session shares a class with the first under no scope. The only escape is degenerate---some
    scope effectively a single class (one object, one session)---which is just the global scope
    wearing another name. No stack of genuinely partitioned scopes serializes; globality is
    irreducible.

    \begin{corollary}[Linearizability is two systems]\label{cor:linearizability}
        Sequential consistency is the single global order $O(\mathrm{all})$. By
        \cref{thm:globality} that order is irreducible---a single global scope, not assembled
        from cross-cutting proper scopes such as per-object or per-session order---and realizing
        one global class across asynchronous nodes is agreement among them, i.e.\ consensus,
        which appears here as the impossibility of tiling a global order from local ones rather
        than as an imported result. Its closure comes free: $O(\mathrm{all})$ forces
        $\mathit{deps}_{\mathrm{explicit}}$ by $\kappa$ (\cref{lem:o-entails-c}), so the store's
        $C$ does not bear on the order. Sequential consistency is nonetheless a single
        configuration: its order is endogenous, a function of the DAG, the serializer's
        decisions being messages within it.

        Linearizability adds the real-time atom $\mathit{rb} \subseteq \mathit{ar}$, and
        $\mathit{rb}$ is \emph{not} a function of any observer's DAG. Two executions with the
        same DAG but different real-time schedules of their concurrent operations differ in
        $\mathit{rb}$, so no single configuration---which reads only its own DAG---can enforce
        it (\cref{sec:appendixA}). Linearizability is therefore a composite of two
        message-passing systems: a \emph{store} holding the data under any $C$, and a global,
        real-time \emph{serializer}---a mutex, a sequencer, a consensus round---that supplies the
        order and routes clients through it in real time to supply $\mathit{rb}$. The serializer
        is the irreducible second system, and \cref{thm:globality} forbids decomposing it into
        local scopes.
    \end{corollary}

    This recovers the classical distinction. Serializability asks only for equivalence to
    \emph{some} serial order---a single global order with no real-time constraint---and is
    single-observer-shaped: one configuration's $O(\mathrm{all})$ suffices. Linearizability adds
    cross-observer real-time precedence, the second system. The boundary is exactly the
    real-time atom that \cref{sec:appendixA} places outside the single-observer completeness
    fragment.

    Read relativistically, the serializer is a \emph{preferred frame}: $\mathit{rb} \subseteq
    \mathit{ar}$ demands a single global simultaneity---an order on spacelike-separated
    operations that no light cone supplies---so it must be manufactured at a distinguished
    point, not read off any DAG. Clause~(3) of \cref{thm:globality} is the relativity of
    simultaneity in discrete form: a global frame cannot be assembled from local ones.

    \section{The Configuration Glossary}
    \label{sec:glossary}

    We map the LCC configuration space onto the named landscape, taking model names from Viotti
    and Vukoli\'{c}~\cite{viotti2016} rather than reproducing their taxonomy. The correspondence
    is a \emph{coincidence on the standard-safety fragment} (\cref{sec:appendixA}), not a
    containment. Over that fragment---the discipline in which named models are specified---it is
    exact. Off it, neither framework contains the other: LCC adds a liveness axis ($R$, including
    $R(\mathrm{absent})$ and the continuous bound $R(\delta)$) that the safety axioms do not range
    over, since they presuppose returned values; conversely, predicates over arbitrary visibility
    fall outside LCC's finite lattice. The map is many-to-one: several named models share a configuration. Vendor names in
    the Example column are illustrative anchors, not measured behavior.

    The space is
    $C \in \{\mathrm{none}, \mathrm{object}, \mathrm{session}, \mathrm{explicit}\}$,
    $O \in \{\mathrm{trivial}, \pi_{\mathrm{object}}, \mathrm{all}\}$, and
    $R \in \{\mathrm{absent}, \delta, \infty, 0\}$---$48$ combinations, of which $8$ are
    degenerate: at $R(0)$ no fork is held open, so the three $O$ scopes collapse to the
    $O(\mathrm{all})$ representative (\cref{cor:retention-quant}). The $40$ distinct
    configurations are \cref{tab:coverage}, grouped by closure scope; finer application-defined
    $\mathit{deps}$ filters extend the space without exhausting it.

    \begin{xltabular}{\textwidth}{@{}rllYYY@{}}
        \caption{LCC configurations mapped to named consistency models: $C$ the closure filter,
            $O$ the fork-resolution scope, $R$ the waiting bound. Model names are Viotti and
            Vukoli\'{c}'s~\cite{viotti2016}; ``---'' marks a configuration with no established name.
            Examples are illustrative anchors.}\label{tab:coverage} \\
        \toprule
        \# & $O$ & $R$ & Name & V\&V model & Example \\
        \midrule
        \endfirsthead
        \toprule
        \# & $O$ & $R$ & Name & V\&V model & Example \\
        \midrule
        \endhead
        \bottomrule
        \endlastfoot
        \multicolumn{6}{@{}l}{\emph{$C = \mathrm{none}$ (no closure)}} \\
        1 & trivial & absent & No guarantees & Weak consistency & UDP, unreliable mail \\
        2 & trivial & $\infty$ & Eventual delivery & Eventual consistency & Dynamo anti-entropy \\
        3 & trivial & $\delta$ & Bounded eventual & Bounded staleness & QoS networks \\
        4 & $\pi_{\mathrm{object}}$ & absent & Per-object unread log & --- & WAL before recovery \\
        5 & $\pi_{\mathrm{object}}$ & $\infty$ & Per-object ordered & FIFO consistency & Kafka per-partition \\
        6 & $\pi_{\mathrm{object}}$ & $\delta$ & Bounded per-object order & Bounded staleness & Cosmos DB bounded \\
        7 & all & absent & Global unread log & --- & Raft log before apply \\
        8 & all & $\infty$ & Total-order broadcast & TOB / Consistent prefix & Redis replication \\
        9 & all & $\delta$ & Bounded total order & --- & MySQL semi-sync \\
        10 & all & 0 & Instant total order & --- & Network switch \\
        \addlinespace
        \multicolumn{6}{@{}l}{\emph{$C = \mathrm{object}$ (per-object closure)}} \\
        11 & trivial & absent & Per-object snapshot & --- & --- \\
        12 & trivial & $\infty$ & Per-object causal & Per-object causal / Slow memory & --- \\
        13 & trivial & $\delta$ & Bounded per-object causal & --- & --- \\
        14 & $\pi_{\mathrm{object}}$ & absent & Per-object ordered snapshot & --- & --- \\
        15 & $\pi_{\mathrm{object}}$ & $\infty$ & Per-object sequential & Per-object sequential / Coherence & Cassandra per-key \\
        16 & $\pi_{\mathrm{object}}$ & $\delta$ & Bounded per-object sequential & --- & --- \\
        17 & all & absent & Global per-object unread log & --- & --- \\
        18 & all & $\infty$ & Per-object globally ordered & --- & Single-leader-per-shard \\
        19 & all & $\delta$ & Bounded per-object global & --- & --- \\
        20 & all & 0 & Per-object sequential (current) & --- & --- \\
        \addlinespace
        \multicolumn{6}{@{}l}{\emph{$C = \mathrm{session}$ (per-session closure)}} \\
        21 & trivial & absent & Session snapshot & --- & Frozen session \\
        22 & trivial & $\infty$ & Session causal & PRAM / RYW / Monotonic reads & MongoDB causal sessions \\
        23 & trivial & $\delta$ & Bounded session causal & --- & --- \\
        24 & $\pi_{\mathrm{object}}$ & absent & Session per-object snapshot & --- & --- \\
        25 & $\pi_{\mathrm{object}}$ & $\infty$ & Session per-object causal & Monotonic writes / Processor / WFR & --- \\
        26 & $\pi_{\mathrm{object}}$ & $\delta$ & Bounded session per-object causal & --- & --- \\
        27 & all & absent & Global session unread log & --- & --- \\
        28 & all & $\infty$ & Session sequential & --- & --- \\
        29 & all & $\delta$ & Bounded session sequential & --- & --- \\
        30 & all & 0 & Session sequential (current) & --- & --- \\
        \addlinespace
        \multicolumn{6}{@{}l}{\emph{$C = \mathrm{explicit}$ (full causal closure)}} \\
        31 & trivial & absent & Consistent snapshot & --- & Point-in-time backup \\
        32 & trivial & $\infty$ & Causal consistency & Causal consistency & --- \\
        33 & trivial & $\delta$ & Bounded causal & Bounded causal & Geo-distributed causal \\
        34 & $\pi_{\mathrm{object}}$ & absent & Per-object causal snapshot & --- & --- \\
        35 & $\pi_{\mathrm{object}}$ & $\infty$ & Causal per-object order & Causal$+$ & COPS, Eiger \\
        36 & $\pi_{\mathrm{object}}$ & $\delta$ & Bounded causal per-object & --- & --- \\
        37 & all & absent & Global causal unread log & --- & Finalized blockchain \\
        38 & all & $\infty$ & Sequential consistency & Sequential & ZooKeeper \\
        39 & all & $\delta$ & Bounded sequential & Timed serial & --- \\
        40 & all & 0 & Sequential (current) & Sequential & Single-node, mutex \\
    \end{xltabular}

    Configurations \#38 and \#40 are both sequential consistency, differing only in liveness:
    \#38 ($R(\infty)$) defers until its waiting strategy is satisfied, \#40 ($R(0)$) returns its
    current resolved state on every answer. A global order entails full causal closure either
    way ($\kappa$, \cref{lem:o-entails-c}), so the store's $C$ is not what distinguishes these
    rows.

    \paragraph*{Unnamed configurations.} More than half of the $40$ carry no established name. A
    few are deployed systems lacking a taxonomic label (\#9, MySQL semi-synchronous replication;
    \#36, CockroachDB); the rest are points the framework predicts. The most distinctive are the
    degenerate $R(\mathrm{absent})$ ``unread log'' configurations---ordered or closed state that
    is never delivered: \#4, a per-object write-ahead log before recovery; \#7, a globally
    ordered log before apply; \#37, a finalized causal log never applied. These are exactly the
    retained, un-applied logs on which \cref{thm:ratchet}'s permanence is conditional
    (\cref{sec:ratchet}). The
    existence of unnamed configurations is a feature, not a gap: the design space is larger than
    the explored region.

    \section{Retention, Quantified}
    \label{sec:retention-quant}

    \Cref{thm:retention} established that resolving a fork requires retaining causal history.
    Here we quantify the cost---how much an observer must hold while a fork is open---and the
    resource trade-off it forces between the ordering and closure scopes.

    \begin{corollary}[What a deferred fork must retain]\label{cor:retention-quant}
        While a fork is held open, an observer must retain the unresolved \emph{frontier}: the
        concurrent candidates together with the causal tags that mark them concurrent. This is a
        multi-valued register, pure metadata---there is no single value to keep---until a
        resolution collapses it. Two independent quantities bound it. Its \emph{cardinality}:
        over a window in which forks stay unresolved for up to $\tau$ at per-class write rate
        $\lambda$, $\Omega(\lambda\tau)$ messages per resolution class, and $\Omega(k\lambda\tau)$
        across $k$ objects under $O(\pi_{\mathrm{object}})$. Its \emph{encoding}:
        $\Omega(\min\{n,s\}\cdot\lg k)$ bits of causal tag per message after $k$ operations with
        $n$ replicas and $s$ objects. Both vanish at $R(0)$ and at $O(\mathrm{trivial})$ with
        confluent $F$.
    \end{corollary}

    \begin{proof}
        By \cref{thm:retention} the distinguishing history of an unresolved fork cannot be
        discarded before resolution, so every concurrent candidate in a class survives in the
        frontier; at per-class rate $\lambda$ over a window of length $\tau$ this is
        $\Omega(\lambda\tau)$ per class, and $\Omega(k\lambda\tau)$ across $k$
        independently-forking objects under $O(\pi_{\mathrm{object}})$---the version count a
        global-minimum-timestamp protocol carries~\cite{bottcher2019,lee2016}. The per-message
        encoding bound is that of Attiya, Ellen, and Morrison~\cite{attiya2015}, extending
        Charron-Bost's vector-clock lower bound~\cite{charronbost1991}. At $R(0)$ no fork is held
        open, and at $O(\mathrm{trivial})$ with confluent $F$ none is opened, so the frontier is
        empty in both.
    \end{proof}

    \begin{corollary}[$C$ and $O$ are inversely constrained]\label{cor:inverse}
        For a fixed retention budget---a bound on the causal metadata an observer keeps
        (\cref{cor:retention-quant})---the achievable ordering scope $O$ is limited: a wider $O$
        entangles more forks and so demands more retained history---Burgess and Gerlits's
        ``wider catchment area, coarser temporal uncertainty''~\cite{burgess2022}, here derived
        rather than observed---while fixing $O$ fixes the minimum metadata needed to resolve its
        forks. This is the \emph{resource} counterpart of
        the \emph{logical} entailment $O \Rightarrow C$ (\cref{lem:o-entails-c}): the entailment
        runs one way---ordering supplies closure, never the reverse---whereas the retention
        budget is a shared resource that $C$ and $O$ both draw on.
    \end{corollary}

    \begin{proof}
        By \cref{cor:retention-quant} the metadata an observer must hold grows with the forks its
        ordering scope keeps open: $O(\mathrm{all})$ co-resolves every concurrent pair into one
        global frontier, $O(\pi_{\mathrm{object}})$ only same-object pairs, $O(\mathrm{trivial})$
        none. A fixed budget therefore caps the realizable $O$, and a chosen $O$ sets the floor on
        the metadata. The constraint is mutual, unlike the entailment of \cref{lem:o-entails-c},
        which is directional.
    \end{proof}

    \section{Related Work}
    \label{sec:related}

    Viotti and Vukoli\'{c}~\cite{viotti2016} surveyed more than fifty consistency models;
    Burckhardt~\cite{burckhardt2014} formalized consistency as predicates over abstract
    executions---the discipline our completeness claim is stated against (\cref{sec:appendixA}).
    Almeida~\cite{almeida2024} composes consistency from binary axioms and derives the CLAM
    impossibility, grounding it, as we do, in a relativistic ``physical happened-before'' over
    light cones---the closest existing framework to ours. The two answer different questions:
    CLAM proves one impossibility boundary, whereas LCC equips the configuration space with
    structure---a decidable readability order with a factoring dichotomy (\cref{thm:dichotomy})
    and the Detection\,$=$\,Prevention bound (\cref{thm:detect-prevent})---and separates
    return-value semantics ($F$, with availability as its behaviour when a requirement cannot be
    met) from the consistency parameters $(C, O)$, whose asymmetric coupling---ordering entails
    closure (\cref{lem:o-entails-c})---is the structural fact a flat parameter space does not
    isolate.

    \paragraph*{The model LCC formalizes.}
    LCC is a formalization of the observer-relative consistency model of Burgess and
    Gerlits~\cite{burgess2022}. Their account is observer-relative from the outset---``current
    value'' is ``an illusion of a single observer's timeline''---and is built on past and future
    \emph{cones} in which a transaction may depend only on its past. Their two resolution
    regimes---\emph{time separation} (locks; all futures collapse into one shared version) and
    \emph{space separation} (private ``Many Worlds'' branches, merged by policy and
    garbage-collected after a horizon)---are the endpoints of our $O$ axis, and the horizon at
    which branches are reclaimed is the retention bound (\cref{thm:retention}). Resolution is the
    receiver's prerogative---their Downstream Principle, ``the receiver of data is always the
    arbiter of its interpretation''---which is $F$ evaluated at the observer; and their
    qualitative tradeoff, ``the wider the catchment area, the coarser the temporal
    uncertainty,'' is our inverse constraint (\cref{cor:inverse}). They read the partition
    tradeoff as ambiguity from incomplete semantics and sidestep the asynchronous-consensus
    impossibility by synchronous coupling---placements made precise by the single-observer
    boundary (\cref{sec:boundary}) and the clock bound (\cref{thm:clockbound}). Their
    hierarchical clock is one realization---a single point in the $(C, O, R, F)$ space---and the
    system that motivated the formalization; we formalize the consistency model, not the clock.

    \paragraph*{Relativistic distributed systems.}
    Schwarz and Mattern~\cite{schwarz1994} drew the analogy between vector time and Minkowski
    causal structure; Gilbert and Golab~\cite{gilbert2014} developed the formal correspondence,
    defining relativistic linearizability. The causal-set program of Bombelli et
    al.~\cite{bombelli1987} proposes spacetime as a locally finite partial order---the structure
    of a causal DAG---and Lamport~\cite{lamport1978} credited special relativity for the
    happened-before relation; Burgess~\cite{burgess2014} models systems as autonomous agents
    whose consistency is local. LCC's light-cone framing builds on this lineage and sharpens an
    observation it shares with it: the classical impossibilities presuppose a \emph{preferred
    frame}---a single global simultaneity over operations that are pairwise spacelike-separated.
    A global order respecting real time is exactly such a frame, and \cref{sec:boundary} shows
    no single observer's light cone supplies one, so it must be manufactured by a serializer
    (clause~(3) of \cref{thm:globality} is the relativity of simultaneity in discrete form).
    Borrill~\cite{borrill2026category,borrill2026flp,borrill2026cap} independently locates a
    shared hidden assumption---forward-in-time-only information flow---beneath the
    asynchronous-consensus impossibility, the partition tradeoff, and the Two Generals problem,
    and circumvents it by engineering a bisynchronous substrate. LCC locates the assumption in
    the same place but treats it differently: not in the link layer, but in the observer's
    inability to see the cause of its sub-DAG's incompleteness. Borrill changes the substrate;
    LCC analyzes the frame. We state this as an observation, not as a unification of the
    impossibilities.

    \paragraph*{Parameterized consistency.}
    Yu and Vahdat's TACT~\cite{yu2002} measures quantitative deviation from strong consistency
    along three axes---\emph{how far} from linearizability, where LCC specifies \emph{which}
    structural property holds. Golab~\cite{golab2018} formally proved the
    latency--consistency tradeoff for normal operation; the LAW theorem~\cite{katsarakis2025}
    shows local reads are impossible under linearizable asynchronous replication; and Bailis et
    al.'s Highly Available Transactions~\cite{bailis2014} map the isolation levels achievable
    without coordination---in our coordinates the confluent fragment, which the readability order
    (\cref{sec:readability}) refines from transactions to arbitrary configurations.

    \paragraph*{Impossibility results.}
    Gilbert and Lynch~\cite{gilbert2002} proved the partition/availability impossibility;
    Fischer, Lynch, and Paterson~\cite{fischer1985} the impossibility of deterministic
    asynchronous consensus with one crash. Attiya, Ellen, and Morrison~\cite{attiya2015} proved
    that causally and eventually consistent stores need unbounded message size for multi-valued
    registers---the encoding bound of \cref{cor:retention-quant} and the audit cost of
    \cref{cor:auditability}. Dolev, Dwork, and Stockmeyer~\cite{dolev1987} characterized the
    minimal synchronism for consensus; Attiya, Enea, and Rom\'{a}n-Calvo~\cite{attiya2026} proved
    that a specification admits an available implementation iff it is arbitration-free, which
    constrains $O$ directly; Attiya et al.~\cite{attiya2011}
    that strongly non-commutative operations in linearizable implementations require synchronization
    that cannot be eliminated. Hellerstein~\cite{hellerstein2010} conjectured and Ameloot,
    Neven, and Van den Bussche~\cite{ameloot2013} proved the CALM principle---a computation has a
    coordination-free implementation iff it is monotone; in our terms the coordination-free
    fragment is the confluent one, and the arbiter a non-confluent, real-time order requires is
    conserved across abstraction layers (\cref{sec:boundary}), not eliminated---distinct from the
    retention bound, which the confluent fragment pays regardless~\cite{mahajan2011}. Gibbons and
    Korach~\cite{gibbonskorach1997} proved sequential-consistency checking NP-complete, fixing
    the cost of detecting the missing-order scar (\cref{rem:downgrade}). The retention bound
    (\cref{thm:retention}) is new: fork resolution requires retained causal history, treated as
    self-evident in the database community~\cite{bottcher2019,lee2016} but not previously stated
    as an impossibility for general message-passing systems. Our contribution is to place these
    results in one coordinate system; we make no claim to unify them.

    \paragraph*{Critiques of the partition tradeoff.}
    Brewer~\cite{brewer2012} recast the tradeoff as a partition-mode and recovery model, with
    consistency and availability continuous rather than binary; Abadi~\cite{abadi2012} added
    latency as an orthogonal axis; Kleppmann~\cite{kleppmann2015} faulted the ambiguity of its
    availability notion and its conflation of finite and infinite partitions; Lee et
    al.~\cite{lee2021cal} replaced partition tolerance with a quantitative latency measure. None
    questions the \emph{unity premise}---that the partitioned components remain one logical
    system whose reunification is desired. Helland~\cite{helland2007} showed operationally that
    consistency is scoped to a chosen entity with cross-scope coordination by messages; LCC makes
    that scope choice the subject of the bound itself, and the retention bound is precisely the
    cost of the reunification the premise assumes (\cref{rem:retention-context}).

    \section{Conclusion}
    \label{sec:conclusion}

    We have described consistency as two operators on a single observer's growing sub-DAG: a
    closure filter $C$ fixing which causal dependencies the observer must have seen, and a fork
    resolution $O$ ordering the concurrent messages the filter admits, with a response function
    $F$ reporting and a waiting bound $R$ setting liveness. Ordering entails closure, so the two
    are coupled rather than independent, and the consistency content of a resolution is whether
    it refines causality. The named models fall out as configurations; the retention bound is the
    price of ever resolving a fork; the readability order is a decidable account of which
    configuration can read another; the merge frontier is when two divergent views reconcile, the
    partition tradeoff what it forces when they cannot; and the ratchet and $\kappa$-scars are the
    permanence of crossing a configuration boundary.

    The boundary of the account is the result we take to be sharpest. Linearizability is not
    single-observer consistency but a composite of a store and a global, real-time
    serializer---a preferred frame manufactured at a distinguished point, because the global
    simultaneity it demands is one no light cone supplies. Read this way, the classical
    impossibilities share a presupposition (a preferred frame) and LCC's coordinates make that
    assumption visible without importing any of the results. We have not unified them; a rigorous
    reconstruction of their content within the single-observer frame remains open, as does the
    carve-out boundary (Byzantine and shared-memory behaviour) and the clock construction by
    which a sufficiently accurate physical clock buys a $\kappa$-clean global order
    (\cref{sec:clocks}). LCC formalizes an observer-relative architecture arrived at independently
    as engineering; closing the formal account to the cases that architecture already handles is
    the natural next step.

    \section*{AI Disclosure}

    We used Claude (Anthropic) to assist with proof verification, LaTeX formatting, and adversarial review of arguments. The tool materially affected the entire manuscript. Specifically:
    \begin{itemize}
        \item \emph{Proof verification.} Claude was used to check the proofs and the Burckhardt mapping in Appendix~A for logical gaps, missing cases, and unstated assumptions.
        \item \emph{LaTeX conformance.} Claude was used to convert the source draft into LIPIcs-conforming LaTeX, including theorem environments, cross-references, bibliography entries, and table formatting.
        \item \emph{Adversarial review.} Claude was used to challenge claims, identify weak arguments, surface counterexamples, and stress-test conjectures before submission.
    \end{itemize}

    The authors verified the correctness and originality of all content including references, and take full responsibility for any errors, omissions, or misattributions in the final manuscript.

    \bibliography{references}

\begin{thebibliography}{10}

\bibitem{abadi2012}
Daniel~J. Abadi.
\newblock Consistency tradeoffs in modern distributed database system design:
  {CAP} is only part of the story.
\newblock {\em Computer}, 45(2):37--42, 2012.
\newblock \href {https://doi.org/10.1109/MC.2012.33}
  {\path{doi:10.1109/MC.2012.33}}.

\bibitem{almeida2024}
Paulo~S\'{e}rgio Almeida.
\newblock A framework for consistency models in distributed systems.
\newblock {\em arXiv:2411.16355}, 2024.
\newblock URL: \url{https://arxiv.org/abs/2411.16355}.

\bibitem{ameloot2013}
Tom~J. Ameloot, Frank Neven, and Jan~Van den Bussche.
\newblock Relational transducers for declarative networking.
\newblock {\em Journal of the ACM}, 60(2):15:1--15:38, 2013.
\newblock \href {https://doi.org/10.1145/2450142.2450151}
  {\path{doi:10.1145/2450142.2450151}}.

\bibitem{attiya2015}
Hagit Attiya, Faith Ellen, and Adam Morrison.
\newblock Limitations of highly-available eventually-consistent data stores.
\newblock In {\em ACM Symposium on Principles of Distributed Computing (PODC)},
  pages 385--394, 2015.
\newblock \href {https://doi.org/10.1145/2767386.2767419}
  {\path{doi:10.1145/2767386.2767419}}.

\bibitem{attiya2026}
Hagit Attiya, Constantin Enea, and Enrique Rom\'{a}n-Calvo.
\newblock Arbitration-free consistency is available (and vice versa).
\newblock {\em Proc. ACM Program. Lang.}, 10(POPL):1183--1211, 2026.
\newblock Article 41.
\newblock \href {https://doi.org/10.1145/3776683} {\path{doi:10.1145/3776683}}.

\bibitem{attiya2011}
Hagit Attiya, Rachid Guerraoui, Danny Hendler, Petr Kuznetsov, Maged~M.
  Michael, and Martin Vechev.
\newblock Laws of order: Expensive synchronization in concurrent algorithms
  cannot be eliminated.
\newblock In {\em 38th ACM Symposium on Principles of Programming Languages
  (POPL)}, pages 487--498, 2011.
\newblock \href {https://doi.org/10.1145/1926385.1926442}
  {\path{doi:10.1145/1926385.1926442}}.

\bibitem{bailis2014}
Peter Bailis, Aaron Davidson, Alan Fekete, Ali Ghodsi, Joseph~M. Hellerstein,
  and Ion Stoica.
\newblock Highly available transactions: Virtues and limitations.
\newblock {\em Proceedings of the VLDB Endowment}, 7(3):181--192, 2013.
\newblock \href {https://doi.org/10.14778/2732232.2732237}
  {\path{doi:10.14778/2732232.2732237}}.

\bibitem{bombelli1987}
Luca Bombelli, Joohan Lee, David Meyer, and Rafael~D. Sorkin.
\newblock Space-time as a causal set.
\newblock {\em Physical Review Letters}, 59(5):521--524, 1987.
\newblock \href {https://doi.org/10.1103/PhysRevLett.59.521}
  {\path{doi:10.1103/PhysRevLett.59.521}}.

\bibitem{borrill2026cap}
Paul Borrill.
\newblock Circumventing the {CAP} theorem with open atomic ethernet, 2026.
\newblock arXiv:2602.21182.
\newblock \href {https://arxiv.org/abs/2602.21182} {\path{arXiv:2602.21182}}.

\bibitem{borrill2026flp}
Paul Borrill.
\newblock Circumventing the {FLP} impossibility result with open atomic
  ethernet, 2026.
\newblock arXiv:2602.20444.
\newblock \href {https://arxiv.org/abs/2602.20444} {\path{arXiv:2602.20444}}.

\bibitem{borrill2026category}
Paul Borrill.
\newblock What distributed computing got wrong: The category mistake that
  turned design choices into laws of nature, 2026.
\newblock arXiv:2602.18723.
\newblock \href {https://arxiv.org/abs/2602.18723} {\path{arXiv:2602.18723}}.

\bibitem{bottcher2019}
Jan B{\"o}ttcher, Viktor Leis, Thomas Neumann, and Alfons Kemper.
\newblock Scalable garbage collection for in-memory {MVCC} systems.
\newblock {\em Proc. VLDB Endow.}, 13(2):128--141, 2019.
\newblock \href {https://doi.org/10.14778/3364324.3364328}
  {\path{doi:10.14778/3364324.3364328}}.

\bibitem{brewer2012}
Eric~A. Brewer.
\newblock {CAP} twelve years later: How the ``rules'' have changed.
\newblock {\em Computer}, 45(2):23--29, 2012.
\newblock \href {https://doi.org/10.1109/MC.2012.37}
  {\path{doi:10.1109/MC.2012.37}}.

\bibitem{burckhardt2014}
Sebastian Burckhardt.
\newblock Principles of eventual consistency.
\newblock {\em Foundations and Trends in Programming Languages},
  1(1--2):1--150, 2014.
\newblock \href {https://doi.org/10.1561/2500000011}
  {\path{doi:10.1561/2500000011}}.

\bibitem{burgess2014}
Mark Burgess and Jan~A. Bergstra.
\newblock {\em Promise Theory: Principles and Applications}.
\newblock 2014.

\bibitem{burgess2022}
Mark Burgess and Andras Gerlits.
\newblock Continuous integration of data histories into consistent namespaces.
\newblock 03 2022.
\newblock \href {https://doi.org/10.13140/RG.2.2.17170.53444}
  {\path{doi:10.13140/RG.2.2.17170.53444}}.

\bibitem{charronbost1991}
Bernadette Charron-Bost.
\newblock Concerning the size of logical clocks in distributed systems.
\newblock {\em Information Processing Letters}, 39(1):11--16, 1991.
\newblock \href {https://doi.org/10.1016/0020-0190(91)90055-M}
  {\path{doi:10.1016/0020-0190(91)90055-M}}.

\bibitem{dolev1987}
Danny Dolev, Cynthia Dwork, and Larry Stockmeyer.
\newblock On the minimal synchronism needed for distributed consensus.
\newblock {\em Journal of the ACM}, 34(1):77--97, 1987.
\newblock \href {https://doi.org/10.1145/7531.7533}
  {\path{doi:10.1145/7531.7533}}.

\bibitem{fischer1985}
Michael~J. Fischer, Nancy~A. Lynch, and Michael~S. Paterson.
\newblock Impossibility of distributed consensus with one faulty process.
\newblock {\em Journal of the ACM}, 32(2):374--382, 1985.
\newblock \href {https://doi.org/10.1145/3149.214121}
  {\path{doi:10.1145/3149.214121}}.

\bibitem{gibbonskorach1997}
Phillip~B. Gibbons and Ephraim Korach.
\newblock Testing shared memories.
\newblock {\em SIAM Journal on Computing}, 26(4):1208--1244, 1997.
\newblock \href {https://doi.org/10.1137/S0097539794279614}
  {\path{doi:10.1137/S0097539794279614}}.

\bibitem{gilbert2014}
Seth Gilbert and Wojciech Golab.
\newblock Making sense of relativistic distributed systems.
\newblock In {\em 28th International Symposium on Distributed Computing
  (DISC)}, volume 8784 of {\em LNCS}, pages 361--375, 2014.
\newblock \href {https://doi.org/10.1007/978-3-662-45174-8_25}
  {\path{doi:10.1007/978-3-662-45174-8_25}}.

\bibitem{gilbert2002}
Seth Gilbert and Nancy Lynch.
\newblock Brewer's conjecture and the feasibility of consistent, available,
  partition-tolerant web services.
\newblock {\em SIGACT News}, 33(2):51--59, 2002.
\newblock \href {https://doi.org/10.1145/564585.564601}
  {\path{doi:10.1145/564585.564601}}.

\bibitem{golab2018}
Wojciech Golab.
\newblock Proving {PACELC}.
\newblock {\em ACM SIGACT News}, 49(1):73--81, 2018.
\newblock \href {https://doi.org/10.1145/3197406.3197420}
  {\path{doi:10.1145/3197406.3197420}}.

\bibitem{helland2007}
Pat Helland.
\newblock Life beyond distributed transactions: An apostate's opinion.
\newblock In {\em Proc.\ 3rd Biennial Conf.\ on Innovative Data Systems
  Research (CIDR)}, pages 132--141, 2007.

\bibitem{hellerstein2010}
Joseph~M. Hellerstein.
\newblock The declarative imperative: Experiences and conjectures in
  distributed logic.
\newblock {\em ACM SIGMOD Record}, 39(1):5--19, 2010.
\newblock \href {https://doi.org/10.1145/1860702.1860704}
  {\path{doi:10.1145/1860702.1860704}}.

\bibitem{katsarakis2025}
Antonios Katsarakis, Vasilis Gavrielatos, Emmanouil Giortamis, Pramod Bhatotia,
  Aleksandar Dragojevi\'{c}, Boris Grot, Vijay Nagarajan, and Panagiota
  Fatourou.
\newblock The {LAW} theorem: Local reads and linearizable asynchronous
  replication.
\newblock {\em Proceedings of the VLDB Endowment}, 18(9):2831--2845, 2025.
\newblock \href {https://doi.org/10.14778/3746405.3746411}
  {\path{doi:10.14778/3746405.3746411}}.

\bibitem{kleppmann2015}
Martin Kleppmann.
\newblock A critique of the {CAP} theorem, 2015.
\newblock arXiv:1509.05393.
\newblock \href {https://arxiv.org/abs/1509.05393} {\path{arXiv:1509.05393}}.

\bibitem{lamport1978}
Leslie Lamport.
\newblock Time, clocks, and the ordering of events in a distributed system.
\newblock {\em Communications of the ACM}, 21(7):558--565, 1978.
\newblock \href {https://doi.org/10.1145/359545.359563}
  {\path{doi:10.1145/359545.359563}}.

\bibitem{lee2021cal}
Edward~A. Lee, Soroush Bateni, Shaokai Lin, Marten Lohstroh, and Christian
  Menard.
\newblock Quantifying and generalizing the {CAP} theorem, 2021.
\newblock arXiv:2109.07771.
\newblock \href {https://arxiv.org/abs/2109.07771} {\path{arXiv:2109.07771}}.

\bibitem{lee2016}
Juchang Lee, Hyungyu Shin, Chang~Gyoo Park, Seongyun Ko, Jaeyun Noh, Yongjae
  Chuh, Wolfgang Stephan, and Wook-Shin Han.
\newblock Hybrid garbage collection for multi-version concurrency control in
  {SAP HANA}.
\newblock In {\em Proceedings of the 2016 International Conference on
  Management of Data (SIGMOD)}, pages 1307--1318. ACM, 2016.
\newblock \href {https://doi.org/10.1145/2882903.2903734}
  {\path{doi:10.1145/2882903.2903734}}.

\bibitem{lloyd2011}
Wyatt Lloyd, Michael~J. Freedman, Michael Kaminsky, and David~G. Andersen.
\newblock Don't settle for eventual: Scalable causal consistency for wide-area
  storage with {COPS}.
\newblock In {\em 23rd ACM Symposium on Operating Systems Principles (SOSP)},
  pages 401--416, 2011.
\newblock \href {https://doi.org/10.1145/2043556.2043593}
  {\path{doi:10.1145/2043556.2043593}}.

\bibitem{mahajan2011}
Prince Mahajan, Lorenzo Alvisi, and Mike Dahlin.
\newblock Consistency, availability, and convergence.
\newblock Technical Report TR-11-22, Department of Computer Science, University
  of Texas at Austin, 2011.

\bibitem{schwarz1994}
Reinhard Schwarz and Friedemann Mattern.
\newblock Detecting causal relationships in distributed computations: In search
  of the holy grail.
\newblock {\em Distributed Computing}, 7(3):149--174, 1994.
\newblock \href {https://doi.org/10.1007/BF02277859}
  {\path{doi:10.1007/BF02277859}}.

\bibitem{viotti2016}
Paolo Viotti and Marko Vukoli\'{c}.
\newblock Consistency in non-transactional distributed storage systems.
\newblock {\em ACM Computing Surveys}, 49(1):19:1--19:34, 2016.
\newblock \href {https://doi.org/10.1145/2926965} {\path{doi:10.1145/2926965}}.

\bibitem{yu2002}
Haifeng Yu and Amin Vahdat.
\newblock Design and evaluation of a conit-based continuous consistency model
  for replicated services.
\newblock {\em ACM Transactions on Computer Systems}, 20(3):239--282, 2002.
\newblock \href {https://doi.org/10.1145/566340.566342}
  {\path{doi:10.1145/566340.566342}}.

\end{thebibliography}

    \appendix

    \section{Completeness over Burckhardt's Standard Safety Axioms}
    \label{sec:appendixA}

    We prove the completeness claim that \cref{sec:boundary} and \cref{sec:glossary} invoke:
    over the \emph{standard-safety fragment} of Burckhardt's abstract-execution
    framework~\cite{burckhardt2014}, every consistency predicate maps to an LCC configuration and
    every configuration to a predicate, the two coinciding exactly on that fragment
    (\cref{thm:completeness}); off it, neither contains the other
    (\cref{cor:expressiveness-full}).

    \subsection{Burckhardt's Framework (Summary)}

    An abstract execution is a tuple $A = (E, \mathit{op}, \mathit{rval}, \mathit{rb},
    \mathit{ss}, \mathit{vis}, \mathit{ar})$: $E$ is a finite set of events; $\mathit{op}: E \to
    \mathrm{Op}$ labels each with its operation; $\mathit{rval}: E \to \mathrm{Val}$ assigns a
    return value; $\mathit{rb} \subseteq E \times E$ is the \emph{returns-before} relation
    ($e_1\,\mathit{rb}\,e_2$ if $e_1$'s response precedes $e_2$'s invocation in real time, an
    irreflexive partial order); $\mathit{ss}$ the \emph{same-session} equivalence; $\mathit{vis}$
    the \emph{visibility} relation ($e_1\,\mathit{vis}\,e_2$: $e_2$ can see $e_1$); and
    $\mathit{ar}$ the \emph{arbitration} relation, a total order resolving conflicts. Session
    order is $\mathit{so} \stackrel{\mathrm{def}}{=} \mathit{rb} \cap \mathit{ss}$, and
    happens-before $\mathit{hb} \stackrel{\mathrm{def}}{=} (\mathit{so} \cup \mathit{vis})^+$.
    Consistency models are conjunctions of axioms over these relations---causal visibility
    ($\mathit{hb} \subseteq \mathit{vis}$), causal arbitration ($\mathit{hb} \subseteq
    \mathit{ar}$), single order ($\mathit{vis} = \mathit{ar}$ modulo incomplete operations),
    real-time ($\mathit{rb} \subseteq \mathit{ar}$), and a return-value axiom.

    \subsection{The Standard-Axiom Fragment}
    \label{sec:std-fragment}

    A \emph{standard predicate} is a conjunction of axiom atoms, at most one from each group,
    drawn from Burckhardt's Figure~5.1:
    \begin{itemize}
        \item \textbf{Visibility} (one of): $\textsc{vis-none}$ ($\mathit{vis}$ unconstrained);
        $\textsc{vis-object}$ ($\mathit{vis}$ contains every same-object causal ancestor);
        $\textsc{vis-session}$ ($\mathit{so} \subseteq \mathit{vis}$); $\textsc{vis-causal}$
        ($\mathit{hb} \subseteq \mathit{vis}$).
        \item \textbf{Arbitration} (one of): $\textsc{ar-none}$; $\textsc{ar-object}$
        ($\mathit{ar}$ total within each object); $\textsc{ar-all}$ ($\mathit{ar}$ a single total
        order).
        \item \textbf{Real-time} (optional, \emph{excluded from the completeness claim}):
        $\textsc{rt}$ ($\mathit{rb} \subseteq \mathit{ar}$)---not a single configuration; see the
        real-time paragraph below.
        \item \textbf{Single-order} (optional): $\textsc{single}$ ($\mathit{vis} = \mathit{ar}$ on
        complete operations).
        \item \textbf{Return value} (one of): $\textsc{rval}_{f}$ for $f \in \{\mathrm{latest},
        \mathrm{any\text{-}conc}, \mathrm{computed}, \mathrm{multi}, \mathrm{any}\}$.
    \end{itemize}
    The arbitration group is the \emph{scope} of conflict resolution and maps to $O$. Causal
    arbitration $\mathit{hb} \subseteq \mathit{ar}$ is \emph{not} a scope atom: it constrains
    whether the resolved order \emph{refines} causality, which LCC carries as the
    causal-arbitration regime on $\prec$ (\cref{ass:causal-arb}), orthogonal to the $O$-scope and
    covered both with it and without it (last-writer-wins) by \cref{thm:dichotomy}. Write
    $\mathcal{F}_{\mathrm{std}}$ for this fragment; every named model in the Viotti--Vukoli\'{c}
    survey~\cite{viotti2016} is specified by a member. Completeness is claimed over the
    \emph{safety} predicates---those without $\textsc{rt}$, which is handled separately below as
    it is not a single configuration. Every visibility atom constrains $\mathit{vis}$ at one of
    the four granularities $\{\mathrm{none}, \text{object}, \text{session}, \text{causal}\}$---the
    four values of $C$---and every arbitration atom partitions $E$ into one of
    $\{\text{trivial}, \text{object}, \text{all}\}$. No standard axiom references a finer
    visibility relation; that restriction is what makes a finite $C$ and the three $O$-scopes
    sufficient. Burckhardt's $\mathit{ar}$ is a single global order, subsumed by the maximal
    $O(\mathrm{all})$; session enters through visibility ($\mathit{so} \subseteq \mathit{vis}$,
    i.e.\ $C(\mathrm{session})$), with no corresponding arbitration atom.

    \begin{remark}[Scope of the claim]\label{rem:completeness-scope}
        Completeness is over the safety predicates of $\mathcal{F}_{\mathrm{std}}$---the
        discipline in which named consistency models are specified---where every visibility
        constraint is one of the four $C$-scopes. A predicate that constrains $\mathit{vis}$ at a
        finer granularity falls outside the fragment, and outside the finite $(C, O, R)$ lattice
        this paper characterises.
    \end{remark}

    \subsection{The Mapping}

    \begin{definition}[Execution correspondence]\label{def:theta}
        An LCC execution $\chi$ is a causal DAG $G = (M, E)$ with a resolved order $\prec$ (a
        linear extension of $G$) and an assignment of each message to a session. It induces the
        abstract execution $\Theta(\chi) = (E_{\mathrm{B}}, \mathit{op}, \mathit{rval},
        \mathit{rb}, \mathit{ss}, \mathit{vis}, \mathit{ar})$ on $E_{\mathrm{B}} = M$ by:
        $\mathit{ar} = {\prec}$; $e_1\,\mathit{vis}\,e_2$ iff $e_1$ is in the seen-set of $e_2$'s
        issuing observer at the instant $e_2$ is created; $\mathit{ss}$ the same-session relation;
        $\mathit{rb}$ the real-time response-before-invocation order; and $\mathit{rval}$ the
        value $F$ returns. Conversely every abstract execution $A$ over $E$ equals $\Theta(\chi)$
        for some $\chi$ (put each event on its session timeline, take $\prec = \mathit{ar}$ and
        $\mathit{deps}$ read off $\mathit{vis}$); $\Theta$ is a bijection on executions over a
        fixed $E$. We write $A \models \sigma$ for ``$A = \Theta(\chi)$ with $\chi$ satisfying
        configuration $\sigma$,'' and $\llbracket \sigma \rrbracket = \{A : A \models \sigma\}$.
    \end{definition}

    \paragraph*{Mapping 1: $\mathit{vis} \to C(\mathit{deps})$.}
    \begin{center}
        \begin{tabular}{@{}ll@{}}
            \toprule
            Burckhardt $\mathit{vis}$ axiom                           & LCC $C(\mathit{deps})$ \\
            \midrule
            $\mathit{vis}$ unconstrained                              & $C(\mathrm{none})$     \\
            $\mathit{so} \subseteq \mathit{vis}$ (session visibility) & $C(\mathrm{session})$  \\
            $\mathit{hb} \subseteq \mathit{vis}$ (causal visibility)  & $C(\mathrm{explicit})$ \\
            $\mathit{vis}$ restricted to same object                  & $C(\mathrm{object})$   \\
            \bottomrule
        \end{tabular}
    \end{center}
    In Burckhardt, $e_1\ \mathit{vis}\ e_2$ means ``$e_2$ sees $e_1$.'' In LCC, $m_1 \in
    \mathit{deps}(m_2) \wedge \mathrm{seen}(n, m_2, t) \implies \mathrm{seen}(n, m_1, t)$:
    $\mathit{vis}$ constrains what must be observed before an observation is valid, and every
    $\mathit{deps}$ filter defines such a constraint.

    \begin{lemma}[Visibility collapse]\label{lem:vis-c}
        Under $\Theta$ (\cref{def:theta}), each visibility atom is equivalent to the matching
        closure scope: for every execution $\chi$ and $X \in \{\mathrm{none}, \mathrm{object},
        \mathrm{session}, \mathrm{causal}\}$,
        \[
            \Theta(\chi) \models \textsc{vis-}X \iff \chi \models C(\mathit{deps}_X),
        \]
        with $\mathit{deps}_{\mathrm{causal}} = \mathit{deps}_{\mathrm{explicit}}$. The one
        non-causal standard visibility pattern---prefix-closed visibility, in which each observer
        sees a $\prec$-prefix---is $C(\mathit{deps}) \wedge O(\mathrm{all})$, and conversely.
    \end{lemma}

    \begin{proof}
    ($\Rightarrow$)
        Suppose $\Theta(\chi) \models \textsc{vis-}X$, i.e.\ $R_X \subseteq
        \mathit{vis}$, where $R_X$ is the corresponding ancestor relation (empty; same-object
        causal ancestors; $\mathit{so}$; $\mathit{hb}$). Under $\Theta$, $e_1\,\mathit{vis}\,e_2$
        holds iff $e_1$ is in $e_2$'s observer's seen-set when $e_2$ is created. So $R_X \subseteq
        \mathit{vis}$ asserts that every $R_X$-ancestor of an observed event is
        observed---precisely downward-closure of every seen-set under $\mathit{deps}_X$, the
        $C(\mathit{deps}_X)$ constraint. ($\Leftarrow$) $C(\mathit{deps}_X)$-closure of every
        seen-set yields $R_X \subseteq \mathit{vis}$ under the same identification. For
        prefix-closed visibility: under $O(\mathrm{all})$ every seen-set is a $\prec$-prefix
        (\cref{def:adm}) and $C(\mathit{deps})$ makes it $\mathit{deps}$-closed, so $\mathit{vis}$
        is exactly a $\mathit{deps}$-closed $\prec$-prefix; conversely such a $\mathit{vis}$ is an
        $O(\mathrm{all})$-admissible, $C$-closed cut.
    \end{proof}

    \emph{Detail.} Session order $\mathit{so} = \mathit{rb} \cap \mathit{ss}$ makes same-session
    operations always returns-before ordered, but a session is not in general a single observer:
    by Creator Visibility (\cref{ax:creator}) a separate connection, thread, or process is a
    distinct observer, so a session outliving a reconnection spans observers. Creator Visibility
    supplies an observer's \emph{own} writes only---single-observer read-your-writes---not a
    same-session predecessor produced by another of the session's observers. The guarantee
    $\mathit{so} \subseteq \mathit{vis}$ is therefore the work of $C(\mathrm{session})$ (or any
    stronger closure): a non-trivial consistency level, not a free property of locality
    (\cref{ax:creator}), met in the $C$ coordinate without $R(\delta)$. The cross-session
    real-time guarantee $\mathit{rb} \subseteq \mathit{ar}$ is not a single configuration but a
    composite of two message-passing systems, treated next.

    \paragraph*{Mapping 2: $\mathit{ar} \to O(\pi)$.}
    \begin{center}
        \begin{tabular}{@{}ll@{}}
            \toprule
            Burckhardt $\mathit{ar}$ axiom        & LCC $O(\pi)$               \\
            \midrule
            $\mathit{ar}$ absent (no arbitration) & $O(\mathrm{trivial})$      \\
            $\mathit{ar}$ per-object total order  & $O(\pi_{\mathrm{object}})$ \\
            $\mathit{ar}$ global total order      & $O(\mathrm{all})$          \\
            \bottomrule
        \end{tabular}
    \end{center}
    The translation is direct: $\mathit{ar}$ restricted to an equivalence class $k$ corresponds
    to $O(\pi)$ with $\pi$ mapping events to $k$. Burckhardt's $\mathit{ar}$ is a single global
    order; LCC's $O$ reads as a \emph{fork-entanglement} scope (\cref{sec:fork}), of which a
    single global order is the maximal value $O(\mathrm{all})$. The three granularities---none,
    per-object, global---are exactly $O(\mathrm{trivial})$, $O(\pi_{\mathrm{object}})$,
    $O(\mathrm{all})$; session is not among them, entering through visibility instead.

    \paragraph*{The real-time axiom $\mathit{rb} \subseteq \mathit{ar}$ is out of scope.}
    Burckhardt's $\mathit{rb}$ is a real-time order: $e_1\ \mathit{rb}\ e_2$ when $e_1$'s response
    precedes $e_2$'s invocation. Realising $\mathit{rb} \subseteq \mathit{ar}$ takes a second
    message-passing system: two executions with the same DAG but different real-time schedules of
    concurrent operations differ in $\mathit{rb}$, so no single $(C, O, R, F)$
    configuration---which reads only its own DAG---can enforce it. It is a \emph{composite}---the
    $O(\mathrm{all})$ order together with a separate coordination system that distributes the
    ordering and routes clients through it (a sequencer, or an accurate clock with
    commit-wait)---and a composite of more than one message-passing system is more than the
    single system LCC describes (\cref{sec:boundary}). The completeness claim is therefore over
    the standard \emph{safety} axioms; the real-time axiom is deferred.

    \paragraph*{Mapping 3: RVAL $\to F$.}
    \begin{center}
        \begin{tabular}{@{}ll@{}}
            \toprule
            Burckhardt RVAL                          & LCC $F$                              \\
            \midrule
            RVAL returns last write in $\mathit{ar}$ & $F_{\mathrm{latest}}$                \\
            RVAL returns any concurrent write        & $F_{\mathrm{any\text{-}concurrent}}$ \\
            RVAL returns deterministic merge         & $F_{\mathrm{computed}}$              \\
            RVAL returns set of concurrent values    & $F_{\mathrm{multi}}$                 \\
            RVAL unconstrained                       & $F_{\mathrm{anything}}$              \\
            \bottomrule
        \end{tabular}
    \end{center}
    RVAL and $F$ operate on the same inputs (visible operations and an ordering) and produce the
    same output; the mapping is direct.

    \subsection{Completeness Proof}

    Let $\varphi$ send each standard axiom atom to its parameter constraint:
    $\varphi(\textsc{vis-}X) = C(\mathit{deps}_X)$; $\varphi(\textsc{ar-}Y) = O(\pi_Y)$ (with
    $\textsc{ar-all} \mapsto O(\mathrm{all})$, $\textsc{ar-none} \mapsto O(\mathrm{trivial})$);
    $\varphi(\textsc{single}) = O(\mathrm{all})$ paired with the matching $C$; and
    $\varphi(\textsc{rval}_f) = F_f$. The real-time atom $\textsc{rt}$ is excluded, as above.

    \begin{lemma}[Per-axiom equivalence]\label{lem:per-axiom}
        For every standard safety axiom atom $\alpha$ and every execution $\chi$, $\Theta(\chi)
        \models \alpha \iff \chi \models \varphi(\alpha)$.
    \end{lemma}

    \begin{proof}
        The arbitration and return-value atoms are immediate: under $\Theta$, $\mathit{ar}$
        \emph{is} $\prec$ and $\mathit{rval}$ \emph{is} $F$'s output, so a constraint on
        $\mathit{ar}$ (a per-class total order) or on $\mathit{rval}$ transfers verbatim to the
        identical object $O(\pi_Y)$ or $F_f$. The visibility atoms are \cref{lem:vis-c}.
        $\textsc{single}$ ($\mathit{vis} = \mathit{ar}$ on complete operations) holds under
        $\Theta$ iff every observed prefix equals an initial segment of $\prec$, which is exactly
        $O(\mathrm{all})$ admissibility together with the closure that pins $\mathit{vis}$ to that
        prefix (\cref{lem:vis-c}, prefix case).
    \end{proof}

    \begin{lemma}[Conjunction]\label{lem:conjunction}
        For standard safety atoms $\alpha_1, \dots, \alpha_k$ (at most one per group),
        $\Theta(\chi) \models \bigwedge_i \alpha_i \iff \chi \models \bigwedge_i
        \varphi(\alpha_i)$, and $\bigwedge_i \varphi(\alpha_i)$ is a single configuration
        $\sigma$.
    \end{lemma}

    \begin{proof}
        Each group constrains a \emph{distinct} primitive of $\chi$: visibility the seen-sets
        ($C$), arbitration the resolved order ($O$), return value the selector ($F$). The
        $\varphi$-images therefore land in distinct coordinates of $\sigma = F(C, O, R)$ and
        never conflict, so their conjunction is one configuration and \cref{lem:per-axiom}
        applies coordinatewise. The one derived relation to check is session order $\mathit{so} =
        \mathit{rb} \cap \mathit{ss}$, which enters only through $\textsc{vis-session}$
        ($\mathit{so} \subseteq \mathit{vis}$): it is met by $C(\mathrm{session})$ in the $C$
        coordinate (\cref{ax:creator}: session visibility is a consistency level, not free), so it
        never conflicts with the $O$ or $F$ atoms. No other pair of standard safety atoms shares a
        derived relation.
    \end{proof}

    \begin{theorem}[Completeness over the standard safety axioms]\label{thm:completeness}
        For every \emph{safety} predicate $P \in \mathcal{F}_{\mathrm{std}}$
        (\cref{sec:std-fragment}, one without the real-time atom $\textsc{rt}$) the configuration
        $\Phi(P) = \bigwedge_{\alpha \in P} \varphi(\alpha)$ satisfies
        \[
            \llbracket P \rrbracket = \llbracket \Phi(P) \rrbracket
            \qquad\text{(exact equivalence over the safety fragment).}
        \]
        The map $P \mapsto \Phi(P)$ is total on the safety predicates and not injective (distinct
        predicates may share a configuration); it is not surjective onto the configuration space
        (\cref{cor:expressiveness-full}).
    \end{theorem}

    \begin{proof}
        By \cref{lem:conjunction}, $\Theta(\chi) \models P \iff \chi \models \Phi(P)$ for every
        $\chi$; since $\Theta$ is a bijection on executions (\cref{def:theta}), the satisfying
        sets coincide, $\llbracket P \rrbracket = \llbracket \Phi(P) \rrbracket$. Totality is by
        construction; the map is non-injective because distinct atoms can constrain the same
        scope, and non-surjective because configurations such as $R(\delta)$ and
        $R(\mathrm{absent})$ have no $\mathcal{F}_{\mathrm{std}}$ preimage
        (\cref{cor:expressiveness-full}). The equivalence is exact \emph{within}
        $\mathcal{F}_{\mathrm{std}}$: a $P$ using arbitrary visibility lies outside it
        (\cref{rem:completeness-scope}).
    \end{proof}

    \begin{corollary}[Coincidence on $\mathcal{F}_{\mathrm{std}}$, incomparability off it]\label{cor:expressiveness-full}
        Over $\mathcal{F}_{\mathrm{std}}$ the map $\Phi$ is an exact equivalence, but not a
        containment in either direction:
        \begin{enumerate}
            \item \emph{(LCC, not in $\mathcal{F}_{\mathrm{std}}$)} $R(\mathrm{absent})$: no
            delivery guarantee. Burckhardt assumes events produce return values, implying
            liveness; no $\mathcal{F}_{\mathrm{std}}$ predicate maps to $R(\mathrm{absent})$.
            \item \emph{(LCC, not in $\mathcal{F}_{\mathrm{std}}$)} The waiting bound $R(\delta)$:
            the standard safety axioms carry no liveness/timing parameter, whereas $R(\delta)$
            parameterises the waiting bound continuously.
            \item \emph{(Burckhardt, not in the finite lattice)} Predicates over arbitrary
            visibility relations lie outside $\mathcal{F}_{\mathrm{std}}$ and have no finite-lattice
            configuration.
        \end{enumerate}
        The frameworks coincide on $\mathcal{F}_{\mathrm{std}}$ and neither contains the other.
    \end{corollary}

    \section{Clock Accuracy for Causal Arbitration}
    \label{sec:clocks}

    \Cref{thm:dichotomy} located physical-timestamp (last-writer-wins) systems by whether their
    arbitration refines causality. The quantitative question that follows---\emph{how good must
    the clocks be?}---acquires a decidable form here, because the $\kappa$ shortcut turns
    ``refines causality'' into a property of the order (\cref{lem:kappa-inv}). The bare
    clock--latency inequality we derive is close to folklore on $\varepsilon$-synchronised clocks
    and hybrid logical clocks; the contribution is not the inequality but the
    \emph{decomposition} it sits inside. LCC separates the cost of \emph{causal} arbitration
    ($\mathit{hb} \subseteq \mathit{ar}$) from the cost of \emph{real-time} arbitration
    ($\mathit{rb} \subseteq \mathit{ar}$), and shows that a commit-wait pays for the latter while
    causal honesty needs far less---usually nothing. The answer is not a number of milliseconds
    but a ratio of clock uncertainty to the relation being protected.

    \begin{lemma}[$\kappa$ is exactly zero causal inversion]\label{lem:kappa-inv}
        For a fixed DAG $G$ and total order $\prec$, the shortcut $(\mathrm{none},\mathrm{all})
        \preceq (\mathrm{explicit},\mathrm{trivial})$ holds iff $\prec$ inverts no causal
        edge---iff $\prec$ is a linear extension of $G$.
    \end{lemma}

    \begin{proof}
        If a causal edge $(j,i)$ is inverted ($i \prec j$), the $\prec$-prefix ending at $i$
        contains $i$ without its parent $j$: not $\mathit{deps}_{\mathrm{explicit}}$-closed, so
        the shortcut fails. If no edge is inverted, $\prec$ is a linear extension, every prefix is
        downward-closed, and the shortcut holds (\cref{lem:o-entails-c}).
    \end{proof}

    \paragraph*{Physical model.} Message $m$ is created at real time $r(m)$ on node $\nu(m)$; its
    timestamp is $\mathrm{ts}(m) = r(m) + b_{\nu(m)}$ for a per-node clock error $b$. A causal
    edge $(j,i)$ requires $r(i) \ge r(j) + L$ for cross-node edges, where $L$ is the link
    latency---causality cannot propagate faster than the network. Same-node causal edges share a
    clock and never invert, so only cross-node edges are at risk, and each has real gap $d_e \ge
    L$. Edge $(j,i)$ inverts iff $\mathrm{ts}(i) \le \mathrm{ts}(j)$, i.e.
    \[
        d_e = r(i) - r(j) \;\le\; b_{\nu(j)} - b_{\nu(i)}.
    \]

    \begin{theorem}[Clock bound for causal arbitration]\label{thm:clockbound}
        Let clock errors be bounded, $|b| \le \varepsilon$, and let equal timestamps be broken
        consistently with causal order (as a hybrid logical clock does). Then timestamp order
        refines causality---and the $\kappa$ regime holds---iff
        \[
            2\varepsilon \;\le\; d_{\min},
        \]
        where $d_{\min}$ is the smallest cross-node causal gap in the execution. Since $d_{\min}
        \ge L$, the condition $2\varepsilon \le L$ is a workload-independent sufficient guarantee:
        \emph{clock uncertainty at most half the fastest causal hop.}
    \end{theorem}

    \begin{proof}
        The largest possible value of $b_{\nu(j)} - b_{\nu(i)}$ is $2\varepsilon$. If
        $2\varepsilon \le d_{\min} \le d_e$ for all $e$, then $d_e \ge b_{\nu(j)} - b_{\nu(i)}$,
        so the inversion inequality $d_e \le b_{\nu(j)} - b_{\nu(i)}$ can hold only with
        equality---a timestamp tie $\mathrm{ts}(i) = \mathrm{ts}(j)$, which the causal tiebreak
        resolves in favour of the earlier event. No \emph{strict} inversion occurs, $\prec$ is a
        linear extension, and $\kappa$ holds by \cref{lem:kappa-inv}. Conversely if $2\varepsilon
        > d_{\min}$, the minimum-gap edge \emph{strictly} inverts ($\mathrm{ts}(i) <
        \mathrm{ts}(j)$) under the worst-case assignment $b_{\nu(j)} = +\varepsilon$, $b_{\nu(i)}
        = -\varepsilon$, which no tiebreak repairs.
    \end{proof}

    \paragraph*{Probabilistic margin.} When $2\varepsilon > d_{\min}$ honesty is not lost outright
    but becomes a probability. For errors drawn independently per operation, uniform on
    $[-\varepsilon, \varepsilon]$, an edge of gap $d$ inverts with probability
    \[
        \Pr[\text{invert} \mid d] = \tfrac{1}{2}\,(1 - \rho)^2,
        \qquad \rho = \tfrac{d}{2\varepsilon} \in [0,1],
    \]
    and $0$ for $\rho \ge 1$ (\cref{fig:clocks}, left). (The worst-case bound of
    \cref{thm:clockbound} holds for any error model with $|b| \le \varepsilon$, including a
    constant per-node bias; the closed form above is the independent-error case. Per-node bias
    instead correlates inversions, so a single skewed node mis-orders all its edges at
    once---worse in the tail, identical at the $2\varepsilon$ threshold.) Requiring an execution
    with $E$ at-risk edges to be honest with probability $\ge 1-\delta$ gives, via the quadratic
    tail, $\varepsilon \le L / \bigl(2(1 - \sqrt{2\delta/E})\bigr)$. The correction is negligible:
    even at $E = 10^6$ and $\delta = 10^{-6}$ the bound is $\varepsilon \le 0.50\,L$. The
    half-the-fastest-hop law is stable across scale.

    \begin{corollary}[The wait is set by the relation you protect]\label{cor:wait}
        A timestamp-ordered system that must enforce $\mathcal{R} \subseteq \mathit{ar}$ for a
        relation $\mathcal{R}$ must insert a minimum pre-response wait
        \[
            w \;=\; \max\bigl(0,\ 2\varepsilon - d_{\min}(\mathcal{R})\bigr),
        \]
        where $d_{\min}(\mathcal{R})$ is the smallest real-time gap between $\mathcal{R}$-related
        operations. For $\mathcal{R} = \mathit{hb}$, $d_{\min}$ is the fastest cross-node
        \emph{causal} hop, bounded below by the link latency $L$, so $w = 0$ whenever
        $2\varepsilon \le L$. For $\mathcal{R} = \mathit{rb}$, the relation orders every
        non-overlapping pair regardless of causality, so $d_{\min}$ can approach $0$ and $w \to
        2\varepsilon$ on every operation.
    \end{corollary}

    Commit-wait is the $\mathit{rb}$-corner of \cref{cor:wait}. Real-time order must sequence even
    causally-unrelated operations that are close in wall-clock time, so it has no propagation
    floor to lean on and must manufacture the gap by waiting $2\varepsilon$. Causal order lives in
    the $\mathit{hb}$-corner, where the floor $L$ is supplied by the network at no cost. A
    real-time system's commit-wait stalls each commit by $\approx 2\varepsilon$ to provide
    external consistency ($\mathit{rb} \subseteq \mathit{ar}$); a system selling only
    \emph{causal} consistency over the same clocks waits $\max(0,\ 2\varepsilon -
    d_{\min}^{\mathit{hb}})$, which across a WAN is zero. The price of real-time order is the
    commit-wait; the price of causal order is usually nothing.

    \begin{remark}[Boundedness is a load-bearing assumption]
        A hard $\kappa$ guarantee requires \emph{bounded} clock uncertainty. Under Gaussian
        (NTP-style) error $b \sim \mathcal{N}(0,\sigma^2)$ the inversion probability is
        $Q\!\bigl(d/(\sigma\sqrt 2)\bigr) > 0$ for every finite gap (\cref{fig:clocks}, left,
        dashed): unbounded clocks make causal honesty arbitrarily likely but never certain, with
        execution-level honesty decaying as $(1 - Q)^E$ over $E$ edges. This is why a bounded
        clock advertises an interval, not a point---and why \cref{thm:clockbound} assumes the
        advertised $\varepsilon$ actually bounds the error: a clock that silently drifts past its
        interval reintroduces inversions exactly as an arbitration violating $\mathit{hb}
        \subseteq \mathit{ar}$ would. Boundedness is an assumption about the clock, not a theorem
        about it.
    \end{remark}

    \begin{remark}[Two mechanisms for the same requirement]\label{rem:currencies}
        The $\kappa$ requirement $\mathit{hb} \subseteq \mathit{ar}$ can be enforced two ways.
        Physical-timestamp systems use clock accuracy: $2\varepsilon \le d_{\min}^{\mathit{hb}}$
        (\cref{thm:clockbound}). Dependency-tracking systems---COPS~\cite{lloyd2011}, and any
        system that orders by the explicit causal edges of $C(\mathrm{explicit})$---use metadata:
        they carry the causal edges explicitly and order by them, so $\varepsilon = 0$ and no
        clock is consulted, incurring the message-size growth of \cref{cor:retention-quant}. The
        two are different mechanisms for the same requirement, and LCC admits both because both
        establish $\mathit{hb} \subseteq \mathit{ar}$. A clock-free, DAG-ordered system is in the
        $\kappa$ regime unconditionally.
    \end{remark}

    \begin{figure}[t]
        \centering
        \includegraphics[width=\linewidth]{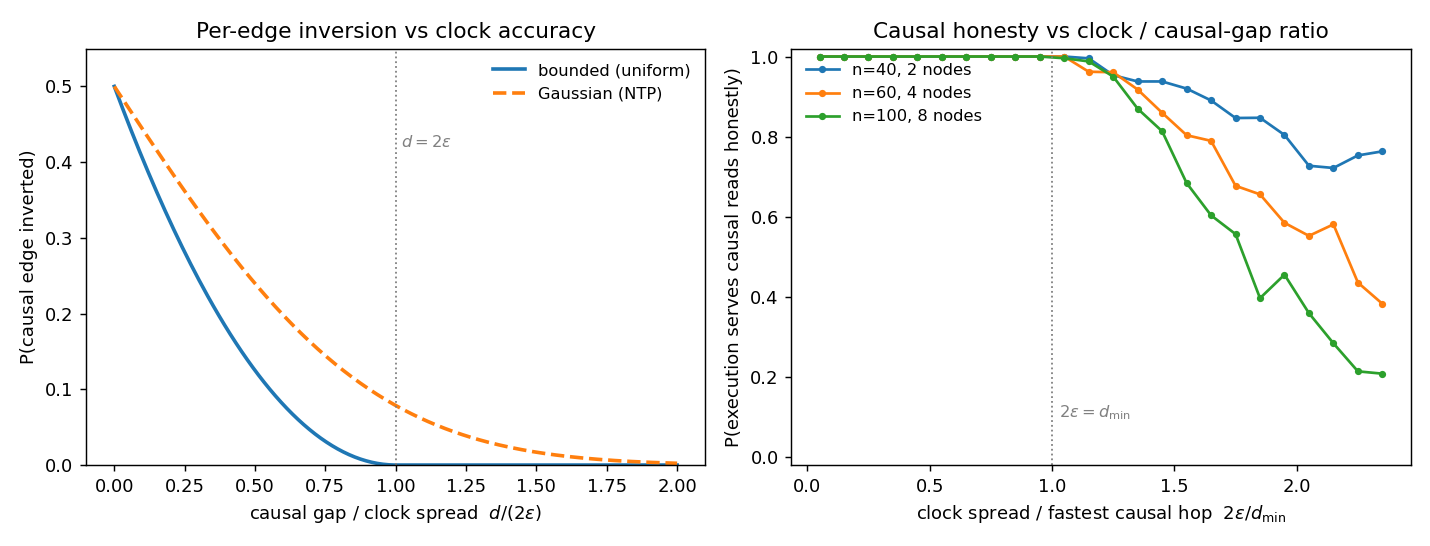}
        \caption{Left: per-edge inversion probability vs the ratio of causal gap to clock spread
            $d/(2\varepsilon)$. Bounded (uniform) error reaches zero at $d = 2\varepsilon$
            (\cref{thm:clockbound}); Gaussian error never does. Right: probability that a whole
            execution serves causal reads honestly, vs $2\varepsilon/d_{\min}$. Honesty is guaranteed
            below $1$ and degrades above, faster for larger systems (more at-risk edges).}
        \label{fig:clocks}
    \end{figure}

\end{document}